\documentclass[aps,prb,twocolumn,amsmath,amssymb,floatfix,nofootinbib,superscriptaddress,longbibliography]{revtex4-2}
\usepackage{graphicx} 
\usepackage{bm} 
\usepackage{hyperref}
\hypersetup{
    colorlinks,
    citecolor=black,
    filecolor=black,
    linkcolor=black,
    urlcolor=black
}
\usepackage{braket,bbold,color,enumerate}
\usepackage{amsthm}

\newcommand{\bs}[1]{\boldsymbol{#1}}

\newtheorem{prop}{Proposition}

\makeatletter
\def\l@subsubsection#1#2{}
\makeatother

\begin{document}

\title{Nonlinear breathers with crystalline symmetries}

\author{Frank Schindler}
\affiliation{Blackett Laboratory, Imperial College London, London SW7 2AZ, United Kingdom}

\author{Vir B. Bulchandani}
\affiliation{Department of Physics, Princeton University, Princeton, New Jersey 08544, USA}

\author{Wladimir A. Benalcazar}
\affiliation{Department of Physics, Emory University, Atlanta, Georgia 30322, USA}

\begin{abstract}
Nonlinear lattice models can support ``discrete breather'' excitations that stay localized in space for all time. By contrast, the localized Wannier states of linear lattice models are dynamically unstable. Nevertheless, symmetric and exponentially localized Wannier states are a central tool in the classification of band structures with crystalline symmetries. Moreover, the quantized transport observed in nonlinear Thouless pumps relies on the fact that -- at least in a specific model -- discrete breathers recover Wannier states in the limit of vanishing nonlinearity. Motivated by these observations, we investigate the correspondence between nonlinear breathers and exponentially localised Wannier states for a family of discrete nonlinear Schr{\"o}dinger equations with crystalline symmetries. We develop a formalism to analytically predict the breathers' spectrum, center of mass and symmetry data, and apply this to nonlinear generalizations of the Su-Schrieffer-Heeger chain and the breathing kagome lattice.
\end{abstract}

\maketitle


\section{Introduction}
Translation-invariant classical lattice models generically exhibit two distinct limits in which they are analytically tractable. In one such limit, the Hamiltonian is a sum of nonlinear on-site terms, and the dynamics consists solely of ``discrete breather'' excitations that are localized to individual lattice sites~\cite{Aubry,Flach1998}. In the other limit, the Hamiltonian is a sum of harmonic couplings between neighbouring lattice sites, and localized wavepackets evolve into delocalized Bloch waves that propagate ballistically in time~\cite{ashcroft2022solid}. One can construct spatially localized Wannier states from Bloch waves by analogy with the Wannier states that arise in the band theory of electrons~\cite{Marzari12}, but such states will be dynamically unstable as they are not eigenmodes. Thus the dynamics between these two solvable limits is qualitatively very different.

For Hamiltonians interpolating between these two extremes, it has been known for some time that discrete breathers and Bloch waves will generally coexist in the low-energy dynamics, provided the breather frequencies are non-resonant with the Bloch band~\cite{Flach1998}. More recently, it has been proposed that there is a direct relation between these two apparently disparate phenonomena in the context of a cubic on-site nonlinearity that, e.g., arises in periodic arrays of photonic waveguides with Kerr nonlinearity or semiclassical regimes of the quantum Bose-Hubbard model. In particular, several experimental and theoretical works on quantized transport in nonlinear Thouless pumps found that discrete breather states continuously recover Wannier states in the limit of vanishing on-site nonlinearity~\cite{jurgensen2021,jurgensen2021chern,mostaan2022quantized,jurgensen2022quantized}. However, to argue for a one-to-one correspondence between breathers and Wannier states, several assumptions must be made: 

\begin{enumerate}[(1)]
\item The Wannier states in question should not only decay exponentially at spatial infinity -- this condition can be satisfied whenever the band structure is not topological~\cite{Brouder07,Soluyanov11,Slager13,Kruthoff17,Po17,Bradlyn17} -- but their exponential decay should set in immediately away from the Wannier center, \emph{i.e.}, they should be supported mostly at a single site.

\item The maximal nonlinear energy contribution to the Hamiltonian should be much smaller than the linear band gap, so that the nonlinear dynamics can be projected into the Wannier basis of a given isolated band to a good approximation~\cite{kevrekidis2009discrete}.

\item This isolated band should be non-degenerate at all momenta.
\end{enumerate}
As long as assumptions (1)--(3) are valid, the effective projected dynamics is accurately captured by the discrete nonlinear Schr{\"o}dinger equation (DNLSE), for which breather solutions have rigorously been proved to exist and be linearly stable~\cite{MacKay1994}. The existence of single-site breathers in the projected dynamics then directly implies the existence of Wannier-like breathers in the original lattice model. 

Here, we investigate the relationship between discrete breathers and Wannier states in situations where assumptions (1)--(3) do not hold and with a particular focus on crystalline symmetries. We consider nonlinear generalizations of tight-binding systems with \emph{molecular bands}, also known as obstructed atomic bands~\cite{Bradlyn17}, whose exponentially localized Wannier states straddle spatially separated sites if they are to preserve the respective crystalline symmetries. By construction, molecular bands maximally violate assumption (1); a textbook example is the Su-Schrieffer-Heeger (SSH) chain~\cite{SSHReview88} with spatial inversion symmetry. In the linear limit, these models are topological in the sense that they cannot be deformed to an (``unobstructed") \emph{atomic} limit, in which the Wannier states are fully localized to individual sites, without breaking the crystalline symmetry or closing a band gap~\cite{Bradlyn17}. 

Our paper develops a general framework that predicts the breather spectrum, linear stability, center of mass, and crystalline symmetry data in an exactly solvable ``molecular" limit, and then applies this to two canonical tight-binding models supplemented by a cubic nonlinearity: the SSH chain in one dimension and the breathing kagome lattice~\cite{Repellin17Kagome} in two dimensions. (Note that in the context of the kagome lattice, the qualifier ``breathing" does not refer to discrete breathers, but rather to an inversion symmetry-breaking spatial modulation of the crystalline lattice.) We then analytically extend our predictions to a small parameter window around each exactly solvable limit and support our claims through numerical simulations far away from these limits. Notably, even though the molecular limits we consider are exactly solvable, they violate both assumptions (1) and (2) in the case of the SSH model, and \emph{all} simplifying assumptions (1)--(3) in the case of the breathing kagome lattice.

While we find that discrete breathers always recover a subset of suitably chosen Wannier states in the limit of vanishing nonlinearity strength, they are subject to qualitatively different symmetry constraints, e.g., due to time-reversal symmetry (TRS), crystalline symmetries, and the gauge freedom in defining multi-band Wannier states. This obstructs a one-to-one correspondence in the multi-band case when assumption (3) is violated. Moreover, when the nonlinearity strength is increased to be of the order of the band gap, such that assumption (2) is violated, we show that a subset of Wannier-like breathers becomes linearly unstable while new and strongly localized breather solutions emerge that are not continously connected to any Wannier states.

We note that some related works exploring connections between band topology and discrete breathers have recently appeared in the literature~\cite{ablowitz2021peierls,jurgensen2021,jurgensen2021chern,mostaan2022quantized,jurgensen2022quantized}. Several of these works refer to ``solitons''; we prefer to reserve the latter term for stable, propagating excitations that scatter elastically off one another~\cite{faddeev2007hamiltonian}, and will instead follow Flach's terminology of ``discrete breathers''~\cite{Flach1998}. Other related work has explored the implications of point-group symmetries for discrete breathers~\cite{crystalref1,crystalref2}, although not from the perspective of Bloch bands and Wannier states adopted below.

\section{Discrete breathers and their linear stability}
In this section, we introduce the main concepts required for analyzing breathers in tight-binding lattices.

\subsection{Crystalline nonlinear Schrödinger equation}
As a generalization of the DNLSE, we study discrete tight-binding models with cubic nonlinearity. These are described by an evolution equation for complex fields $\Phi_a(t) \in \mathbb{C}$,
\begin{equation} \label{eq:nonlinear_tightbinding}
i\frac{\mathrm{d}}{\mathrm{d}t} \Phi_a = \sum_b h_{ab} \Phi_b + g |\Phi_a|^2 \Phi_a,
\end{equation}
where the composite index $a=(\bs{R},\alpha)$ runs over unit cells at lattice position $\bs{R}$ and their internal degrees of freedom $\alpha$, $h_{ab}$ is a translationally invariant and spatially local hopping matrix with an associated band structure, and $g \in \mathbb{R}$ is the nonlinearity strength. We will assume that $g \neq 0$. Physically speaking, the field $\Phi$ can be interpreted as the electric field in periodic arrays of photonic waveguides or as the mean-field wavefunction of a Bose-Einstein condensate. Several recent experiments have realized effective classical dynamics described by Eq.~\eqref{eq:nonlinear_tightbinding} in arrays of photonic waveguides with Kerr nonlinearity, and thereby experimentally demonstrated the existence of breathers and their topological transport across closed adiabatic cycles~\cite{jurgensen2021, jurgensen2022quantized} for suitable choices of hopping matrices $h_{ab}$.


Eq. \eqref{eq:nonlinear_tightbinding} can be written in Hamiltonian form as~\cite{faddeev2007hamiltonian}
\begin{equation}
i\frac{\mathrm{d}}{\mathrm{d}t} \Phi_a = \frac{\partial H}{\partial \Phi_a^*},
\end{equation}
where the Hamiltonian is given by
\begin{equation}
\label{eq:nonlinear_tightbinding_H}
H = \sum_{a,b} \Phi_a^* h_{ab}\Phi_b + \frac{1}{2} g \sum_{a} |\Phi_a|^4.
\end{equation}
Note that $H$ is only real if $h$ is Hermitian, $h = h^\dagger$, in which case the model has a global $U(1)$ symmetry $\Phi_a \rightarrow e^{i \theta} \Phi_a$ \mbox{($\theta \in \mathbb{R}$)} associated with number conservation 
\begin{equation}
\frac{\mathrm{d}}{\mathrm{d}t} \mathcal{N} = 0, \quad \mathcal{N} = \sum_a |\Phi_a|^2.
\end{equation}

Taking the complex conjugate of Eq.~\eqref{eq:nonlinear_tightbinding}, we obtain
\begin{equation}
-i\frac{\mathrm{d}}{\mathrm{d}t} \Phi^*_a = \sum_b h^*_{ab} \Phi^*_b + g |\Phi_a|^2 \Phi^*_a.
\end{equation}
Hence $\Phi^*$ satisfies the time-reversed equation 
\begin{equation}
-i\frac{\mathrm{d}}{\mathrm{d}t} \Phi^*_a = \sum_b h_{ab} \Phi^*_b + g |\Phi_a|^2 \Phi^*_a,
\end{equation}
which is obtained from Eq.~\eqref{eq:nonlinear_tightbinding} by the time-reversal transformation $t \rightarrow -t, \Phi_a \rightarrow \Phi^*_a$, if and only if $h^* = h$ is real.

In this paper, we will assume that $h = h^* = h^T$, \emph{i.e.}, that the dynamics Eq. \eqref{eq:nonlinear_tightbinding} has TRS, is number conserving, and is generated by a real Hamiltonian. We will additionally fix the choice of normalization $\mathcal{N} = 1$, so that breather solutions correspond physically to single-particle excitations, and have vanishing energy density in the thermodynamic limit. Our results can be straightforwardly extended to any $0 < \mathcal{N} < \infty$ by rescaling the effective interaction strength as $g \mapsto g/\mathcal{N}$, which does not qualitatively affect our conclusions.

\subsection{Rotating frame and discrete breathers}
To derive the discrete breather solutions of primary interest in this work, it is useful to change variables to the rotating frame
\begin{equation}
\Phi_a(t) = e^{-i\omega t} \tilde{\Phi}_a(t), 
\end{equation}
for some real frequency $\omega \in \mathbb{R}$ to be determined. In the rotating frame, Eq. \eqref{eq:nonlinear_tightbinding} takes the form
\begin{equation}
\label{eq:rotated_frame}
i \frac{\mathrm{d}}{\mathrm{d}t} \tilde{\Phi}_a = \sum_b \tilde{h}_{ab} \tilde{\Phi}_b + g |\tilde{\Phi}_a|^2 \tilde{\Phi}_a,
\end{equation}
with the ``rotated hopping matrix'' $\tilde{h}$ given by
\begin{equation}
\tilde{h}_{ab} = h_{ab} - \omega \delta_{ab}.
\end{equation}
It will be important below that the rotated equation of motion is itself Hamiltonian, with respect to the ``rotated Hamiltonian''
\begin{equation}
\tilde{H} = \sum_{a,b} \tilde{\Phi}_a^* \tilde{h}_{ab}\tilde{\Phi}_b + \frac{1}{2} g \sum_{a} |\tilde{\Phi}_a|^4.
\end{equation}
Then a concise (but not the most general~\cite{Flach1998}) definition of discrete breather solutions is that they are stationary in the rotating frame,
\begin{equation}
\frac{\mathrm{d}}{\mathrm{d}t} \tilde{\Phi}_a = 0.
\end{equation}
This implies that the discrete breather solution $\tilde{\Phi}_a \equiv \phi_a$ is constant in time when viewed from the rotating frame and satisfies 
\begin{equation} \label{eq:time_indep_DNLSE}
\sum_b (\tilde{h}_{ab} + g |\phi_a|^2 \delta_{ab}) \phi_b = 0.
\end{equation}
Since the latter defines an eigenvalue equation for a real symmetric matrix by TRS, we may set $\phi_{a} \in \mathbb{R}$ without loss of generality. 

We note that breather solutions break translational symmetry as they are localized in space: we will only consider solutions whose envelope $\phi_a$ decays exponentially far away from a fixed center of mass, and whose frequency $\omega$ does not coincide with any eigenfrequency of the linear hopping matrix $h_{ab}$. Correspondingly, and unlike Bloch waves, breather solutions are not parametrized by crystal momentum. Moreover, since Eq.~(1) is nonlinear, it is not possible to superimpose breather solutions centered around different positions to create wavelike solutions with a well-defined momentum.

The remainder of this paper will be concerned with understanding the solutions of Eq. \eqref{eq:time_indep_DNLSE} for various nontrivial hopping matrices $h$. A detailed and self-contained discussion of the linear stability theory of such solutions is given in Appendix \ref{app:stability}. 

\section{Molecular breather limit} \label{sec: molecularlimit}
The band structures given by $h$ in Eq.~\eqref{eq:nonlinear_tightbinding} that we consider in this paper admit a \emph{molecular limit} in which, upon tuning suitable parameters, takes the form
\begin{equation}
h_{ab} = h_{(\bs{R},\alpha),(\bs{R}',\beta)} = \delta_{\bs{R},\bs{R}'} h_{\alpha \beta},
\end{equation}
where $\alpha,\beta$ only range over the degrees of freedom \emph{within each unit cell}, see text below Eq.~\eqref{eq:nonlinear_tightbinding}. This can be seen as a generalization of MacKay and Aubry's ``anti-continuous limit''~\cite{MacKay1994,Aubry} for the standard discrete nonlinear Schr{\"o}dinger equation, some of whose properties we review in Appendix \ref{appendix: dnlse}. We note that unlike for the anti-continuous limit, dynamics in the molecular limit may no longer be integrable; this makes the molecular limit more complicated to analyse in general than the anti-continuous limit.

In the molecular limit, there is no hopping between unit cells. The molecular limit is thus also a dispersionless ``flat band'' limit. However, not all flat band limits necessarily correspond to a molecular limit~\cite{schindler21noncompact}. Similarly, not all atomic bands have a symmetry-preserving molecular limit for a fixed unit cell dimension. Nevertheless, they can always be adiabatically deformed to a molecular limit as long as further orbitals are added to the unit cell~\cite{Else19,schindler21noncompact}. For the sake of simplicity, we will restrict our attention in this paper to the subclass of atomic bands with a molecular limit.

\subsection{Exact solutions}
In the molecular limit, the discrete breather condition Eq.~\eqref{eq:time_indep_DNLSE} reads
\begin{equation} \label{eq:time_indep_DNLSE_molecular_limit}
\omega \phi_{(\bs{R},\alpha)} = \sum_\beta h_{\alpha \beta} \phi_{(\bs{R},\beta)} + g |\phi_{(\bs{R},\alpha)}|^2 \phi_{(\bs{R},\alpha)},
\end{equation}
where $\alpha = 1,2,\ldots,M$ indexes the degree of freedom within each unit cell. The most localized solution that we can construct at the origin, up to translations, is then
\begin{equation} \label{eq: mostlocalized_soln}
    \phi_{(\bs{R},\alpha)} = \delta_{\bs{R},\bs{0}} \phi_\alpha,
\end{equation}
where $\phi_\alpha$ is normalized to $\sum_\alpha |\phi_\alpha|^2 = 1$ and solves the ``reduced'' equations
\begin{equation} \label{eq:time_indep_DNLSE_molecular_limit_reduced}
\omega \phi_{\alpha} = \sum_\beta h_{\alpha \beta} \phi_{\beta} + g |\phi_{\alpha}|^2 \phi_{\alpha}.
\end{equation}

We will study this equation analytically for several molecular-limit model Hamiltonians of interest. Moving away from the molecular limit, these solutions provide good starting points for numerical algorithms for solving Eq. \eqref{eq:time_indep_DNLSE} (see App.~\ref{appendix: numerics}). 

\subsection{Local existence and stability} \label{subsec: molecular_stability}
As discussed in Refs.~\onlinecite{MacKay1994} and \onlinecite{johansson1997existence}, proving existence of discrete breathers in some neighbourhood of the molecular limit requires checking that the dynamics in the molecular limit is anharmonic, i.e. that it cannot be reduced to a system of harmonic oscillators by any canonical transformation of the phase-space coordinates $(p_a,q_a)$. Since this is easiest to show using action-angle variables, which requires both integrability of the molecular limit and knowledge of these variables (neither of which 
is expected for general $h_{\alpha \beta}$), we will simply assume anharmonicity in what follows and rely upon numerical evidence for the existence of breathers instead.

Assuming such local existence holds, a stability analysis can be performed following Appendix \ref{sec: kreinsignature}. The main result of that section is that the linear-response dynamics of small perturbations is determined by a $2N$-by-$2N$ ``stability matrix'' $J \mathcal{L}$ (here $N$ denotes the dimension of $\tilde{h}$ and we refer to Appendix \ref{sec: kreinsignature} for precise definitions of $J$ and $\mathcal{L}$). Purely imaginary eigenvalues of $J \mathcal{L}$ correspond to stable linear perturbations, while eigenvalues of $J \mathcal{L}$ with a non-zero real part correspond to unstable linear perturbations. (Zero eigenvalues are conditionally stable~\cite{kapitula2013spectral} and will be discussed on a case-by-case basis.)

First consider linear stability in the molecular breather limit. For solutions of the form of Eq.~\eqref{eq: mostlocalized_soln}, the stability matrix $J\mathcal{L}$ Eq.~\eqref{eq: stabilityham} decomposes into two types of matrix blocks: a single block of the form
\begin{equation}
\label{eq:stabhambreath}
    (B_{0})_{\alpha \beta} = \begin{pmatrix} 0 & 
    -\tilde{h}_{\alpha \beta} - 3g \phi_{\alpha}^2 \delta_{\alpha \beta}
    \\ 
    \tilde{h}_{\alpha \beta} + g \phi_{\alpha}^2 \delta_{\alpha \beta}
    & 0 \end{pmatrix},
\end{equation}
at lattice position $\bs{R}$, 
and an extensive number $L-1$ of antisymmetric blocks of the form
\begin{equation}
\label{eq:stabhamnonbreath}
    (B_1)_{\alpha \beta} = \begin{pmatrix} 0 & 
    -\tilde{h}_{\alpha \beta} 
    \\ 
    \tilde{h}_{\alpha \beta}
    & 0 \end{pmatrix},
\end{equation}
at lattice positions $\bs{R}' \neq \bs{R}$, where $L$ denotes the number of unit cells in the lattice. As $B_1$ is always anti-Hermitian, only the block $B_0$ at the breather location $\bs{R}$ can give rise to linear instability in the molecular limit. The form of $B_0$ depends on the specific form of the breather under consideration.

The other block $B_1$ does not depend at all on the shape of the breather, but only on its frequency and the linear hopping matrix $h_{\alpha \beta}$. In particular, the spectrum of $B_1$ always lies on the imaginary axis: given eigenvectors $\psi^{(j)}$ of $h_{\alpha \beta}$ with $h \psi^{(j)} = E^{(j)} \psi^{(j)}$, we can diagonalize
\begin{equation} \label{eq: ml_off_site_stability_eigenvectors}
    B_1 \frac{1}{\sqrt{2}}\begin{pmatrix} \psi^{(j)} \\ \pm i \psi^{(j)} \end{pmatrix} = \mp i (E^{(j)} - \omega) \frac{1}{\sqrt{2}} \begin{pmatrix} \psi^{(j)} \\ \pm i\psi^{(j)} \end{pmatrix},
\end{equation}
so that $B_1$ always has eigenvalues $\lambda^{(j)} = \pm i(E^{(j)} - \omega)$. Each non-degenerate eigenvalue of $B_1$ results in an $(L-1)$-fold degenerate eigenvalue of $J\mathcal{L}$. Away from the molecular limit, these eigenvalues will perturbatively recover two (shifted and rescaled) copies of the band structure of the hopping matrix $h_{ab}$. We will further assume the non-degeneracy condition $E^{(j)} \neq \omega$, which generically only fails for finitely many $g \neq 0$ but will be checked below.

Given the latter condition, we can apply the analysis of Sec. \ref{sec: kreinsignature} to explore stability away from the molecular limit. The Krein matrices for eigenvalues of $B_0$ again need to be defined on a case-by-case basis. Assuming no degeneracies within the spectrum of $B_1$ and none between the spectra of $B_0$ and $B_1$ (note that this is the generic case), the Krein matrix corresponding to the highly degenerate stability eigenvalue $\lambda = i(E^{(j)}-\omega)$ is given by the $(L-1) \times (L-1)$ matrix
\begin{equation}
\label{eq:B1KreinMatrix}
    K_{\bs{R}',\bs{R}''}= \delta_{\bs{R}'\bs{R}''}(E^{(j)}-\omega), \quad \bs{R}',\bs{R}'' \neq \bs{R}.
\end{equation}
By Proposition \ref{prop1} and our assumption that $E^{(j)} \neq \omega$, it follows that the eigenvalue $\lambda$ cannot bifurcate in some neighbourhood of the molecular limit. The eigenvalue $\lambda = -i(E^{(j)}-\omega)$ has the same Krein matrix Eq. \eqref{eq:B1KreinMatrix}. As a corollary of this analysis, it follows that collisions between pairs of eigenvalues of $B_1$ with $E^{(j)}-\omega = E^{(k)}-\omega \neq 0$ will have a definite Krein matrix, and therefore cannot give rise to bifurcations. However, collisions between pairs of eigenvalues with $E^{(j)} - \omega = \omega - E^{(k)} \neq 0$, such that 
\begin{equation}
\label{eq: mysterycondition}
E^{(j)} + E^{(k)} = 2\omega
\end{equation}
will have an indefinite Krein matrix. A physical intuition for the condition Eq. \eqref{eq: mysterycondition} is that it implies a nonlinear resonance between the molecular breather state and the flat band in question, i.e. a small denominator in classical perturbation theory~\cite{wayne1994introduction}, which may generate a nonlinear instability when the system is perturbed. A complete stability analysis requires more detailed knowledge of $B_0$ and $B_1$.

\subsection{Symmetry data} \label{sec: symmetry_data_introduction}
Discrete breathers in the molecular limit are further constrained by unitary symmetries of the dynamics. To be specific, provided both the linear and the nonlinear terms of Eq.~\eqref{eq:time_indep_DNLSE} preserve a given crystalline symmetry $U$ (see App.~\ref{appendix: symmetries} for details), the action of $U$ on a discrete breather solution, $\phi \mapsto U \phi$, will always yield another discrete breather solution with the same frequency as $\phi$. 

For instance, the translational symmetry of Eq.~\eqref{eq:time_indep_DNLSE} always gives rise to an extensively large orbit of frequency-degenerate breather solutions that are obtained by replacing $\delta_{\bs{R},\bs{0}} \rightarrow \delta_{\bs{R},\bs{R}'}$ in Eq.~\eqref{eq: mostlocalized_soln}, for $\bs{R}'$ any unit cell in the lattice. 

In this work, we will specifically be interested in the point group symmetries of the underlying lattice, such as spatial inversion symmetry $\mathcal{I}$ or three-fold rotational symmetry $\mathcal{C}_3$. Each point group symmetry operation $U$ generates a cyclic group $C_U = \langle U \rangle$. The action of $C_U$ on a discrete breather solution $\phi$ defines the orbit $C_U[\phi] = \{U^m\phi:m=0,1,\ldots\}$. We note that the size of $C_U[\phi]$ is at most the order of $C_U$, assumed finite. A special case arises when $\phi$ is an eigenvector of $U$, $U\phi = \lambda_U \phi$; then $\lambda_U$ must be a root of unity of degree $|C_U|$.

We call the collection of eigenvalues $\lambda_U$ and/or orbit dimensions $|C_U[\phi]|$ for all point group generators $U$ the \emph{symmetry data} of a given nonlinear breather solution $\phi$. The symmetry data is a discrete object -- it cannot change continuously -- and so it must remain invariant when tuning away from the molecular breather limit until a frequency gap with other breather modes is closed.

\subsection{Wannier states}
In the molecular limit, one possible choice of Wannier states spanning the spectrum of the tight-binding matrix $h_{ab}$ in Eq.~\eqref{eq:nonlinear_tightbinding} is given by
\begin{equation} \label{eq: wannier_molecular_general_def}
    W^{(n)}_{(\bs{R},\alpha)} = \delta_{\bs{R},\bs{0}} W^{(n)}_\alpha,
\end{equation}
where $n$ labels the different Wannier states that we have chosen to be localized to the unit cell at $\bs{R}=\bs{0}$ (the remaining Wannier states are then obtained by lattice translations). Here, $W^{(n)}_\alpha$ are the eigenvectors of the matrix $h_{\alpha \beta}$ in Eq.~\eqref{eq:time_indep_DNLSE_molecular_limit_reduced}. Defined in this way, the molecular-limit Wannier states are eigenstates of $h_{ab}$, whose spectrum consists of a discrete set of flat bands. They are also delta-localized and symmetric by definition, as long as the $W^{(n)}_\alpha$ are chosen in a symmetric fashion. Away from the molecular limit, any choice of exponentially-localized and symmetric Wannier states will in general spread beyond the ``home" unit cell (chosen to be $\bs{R} = \bs{0}$ above), and these Wannier states will cease to be exact eigenstates of $h_{ab}$. However, their topological properties protected by crystalline symmetries~\cite{Slager13,Kruthoff17,Po17,Bradlyn17} will remain the same as those of the states in Eq.~\eqref{eq: wannier_molecular_general_def}: for the molecular band structures considered here, these are the Wannier state centers of mass and the set of irreducible representations of the crystalline symmetry group furnished by the Wannier basis (the ``Wannier state symmetry data").

\begin{figure}[t]
\centering
\includegraphics[width=0.48\textwidth,page=1]{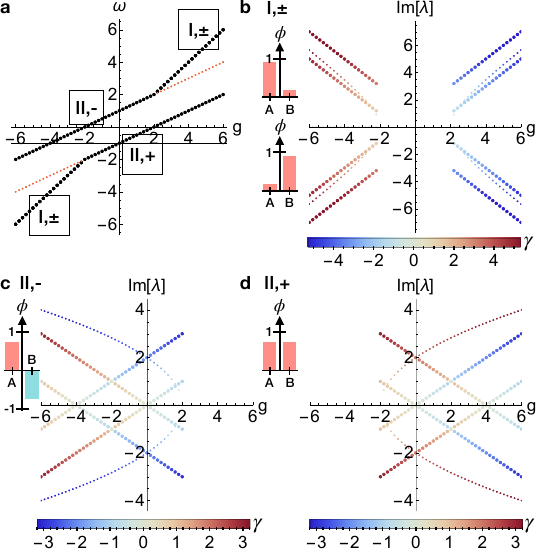}
\caption{Analytical molecular limit breather solutions of the nonlinear SSH chain (Eqs.~\eqref{eq:time_indep_DNLSE},~\eqref{eq: SSH}) for $m = 1, \epsilon = 0$. \textbf{a}~Breather frequencies derived in Sec.~\ref{sec: ssh_molecular_limit}. The frequencies of linearly stable (unstable) breather solutions are shown with thick black (thin red) dots. The linear SSH spectrum at $\omega = \pm m$ is indicated by gray horizontal lines.
\textbf{b}, \textbf{c}, \textbf{d}~Linear stability spectrum $\mathrm{Im}[\lambda]$ and Krein eigenvalues $\gamma$ of the stable breather solutions (the curves terminate when the respective solution becomes unstable). Here, non-degenerate (macroscopically degenerate) stability eigenstates are indicated by thin (thick) dots (see Sec.~\ref{subsec: molecular_stability} for details).
At nonlinearity $|g|>2$, there exist two degenerate stable breathers (solutions $(\mathrm{I},\pm)$) that are mostly localized on either the A ($+$) or B ($-$) site and exchanged by inversion symmetry. In the limit $g \rightarrow \infty$, these modes approach delta-localized breathers (see also App.~\ref{appendix: dnlse}).
At $g < 2$ ($g > -2$) respectively, a further breather mode exists (solutions $(\mathrm{II},\pm)$) that continuously connects to the SSH band at $\omega = m$ ($\omega = -m$). These ``band" breathers share their inversion eigenvalue and center of mass with the Wannier states of the respective linear band.
}
\label{fig: ssh_theory}
\end{figure}

Since we defined the molecular limit Wannier states such that they are eigenstates of $h_{ab}$, it is clear that the breather solutions to Eq.~\eqref{eq:time_indep_DNLSE_molecular_limit} must recover at least a subset of Wannier states in the linear limit $g \rightarrow 0$, assuming that the breathers in question exist for arbitrarily small $|g|$. More specifically,
\begin{equation}
    \lim_{g \rightarrow 0 } \phi^{(\tilde{n})}_{(\bs{R},\alpha)} = W^{(\tilde{n})}_{(\bs{R},\alpha)}
\end{equation}
must hold for $\tilde{n}$ indexing a suitably chosen set of breathers and Wannier states.
We note this equality is only strictly true in the molecular limit. However, because it holds in this limit, the breather solutions must share their topological properties with a subset of Wannier states also away from the molecular limit. In particular, in the case of molecular bands, the breather center of mass must agree with that of the respective Wannier states, and the breather and Wannier state symmetry data must be compatible. This relationship can be used to constrain the breather spectrum based on knowledge of the linear band structure alone. However, it is important to note that this relationship does not in general allow us to predict the breather solutions from the linear limit: First, as we will show in Sec.~\ref{sec: ssh_overarching}, some breather solutions only exist in a window of nonlinearity strengths $g$ that does not include the linear limit $g = 0$. In this case, the linear limit imposes no constraints on the breather solutions themselves, but it can still affect their linear stability. Second, as we show in Sec.~\ref{sec: kagome}, even those breathers that do persist at arbitrarily small $|g|$ in general only recover a proper subset of Wannier states. 

As we will demonstrate in the following, both these inadequacies of the linear limit in predicting the full breather spectrum are remedied by solving the breather problem in the molecular limit and then extrapolating beyond this limit using the Krein stability theory discussed in Sec.~\ref{sec: kreinsignature}.

\section{Nonlinear Su-Schrieffer-Heeger chain} \label{sec: ssh_overarching}
We now turn to a simple example for which the molecular limit is useful, namely a molecular band structure with nonlinearity in 1D. 

\subsection{Tight-binding matrix and symmetries}
The Su-Schrieffer-Heeger (SSH) model~\cite{SSH79} has two sites per unit cell $\alpha = A,B$, and corresponds to the real-space hopping matrix
\begin{equation}
\begin{aligned}
&h_{(R,\alpha),(R',\beta)} = \\ 
&\begin{pmatrix}
0 & -m \delta_{R,R'} - \epsilon \delta_{R+1,R'} \\ 
-m \delta_{R,R'} - \epsilon \delta_{R-1,R'} & 0
\end{pmatrix}_{\alpha \beta},
\end{aligned}
\label{eq: SSH}
\end{equation}
where $m\geq 0$ and $\epsilon \geq 0$ represent intra and inter-unit cell hopping, respectively. The linear system is solved by a Fourier transform to momentum space, which yields the Bloch matrix
\begin{equation} \label{eq: linearblochssh}
h(k)_{\alpha \beta} = -(m + \epsilon \cos k)\sigma^x_{\alpha \beta} - \epsilon \sin k \sigma^y_{\alpha \beta}.
\end{equation}
This has two gapped bands as long as $|m| \neq |\epsilon|$. For $m=0$ ($\epsilon=0$), these become perfectly flat at frequencies $\omega = \pm \epsilon$ ($\omega = \pm m$). We now study breather modes in the nonlinear system defined by letting $h_{(R,\alpha),(R',\beta)}$ be the SSH hopping matrix in Eq.~\eqref{eq:nonlinear_tightbinding}, which we will refer to as the ``nonlinear SSH chain''. This problem is equivalent to fixed-time slices of the nonlinear Thouless pump that has recently enjoyed both experimental and theoretical attention~\cite{jurgensen2021,jurgensen2021chern,jurgensen2022quantized,mostaan2022quantized}.

Beyond lattice translation symmetry, the SSH model exhibits crystalline inversion symmetry $\mathcal{I}_{(R,\alpha),(R',\beta)} = \delta_{R',-R} \sigma^x_{\alpha \beta}$ that satisfies $\mathcal{I}^2 = \mathbb{1}$. We confirm in App.~\ref{appendix: symmetries} that the nonlinear part of the breather equation Eq.~\eqref{eq:time_indep_DNLSE} also preserves this inversion symmetry. Correspondingly, there are three distinct possibilities for the breather symmetry data: the orbit under inversion symmetry $\{\phi,\mathcal{I}\phi\}$ either forms a pair of linearly independent states $\{\phi,\phi'\}$, or takes the form $\{\phi, \pm \phi\}$, corresponding to an inversion eigenvalue $\lambda^{(\mathcal{I})}=\pm 1$. We refer to the latter cases as inversion even and odd respectively.

\begin{figure}[t]
\centering
\includegraphics[width=0.48\textwidth,page=1]{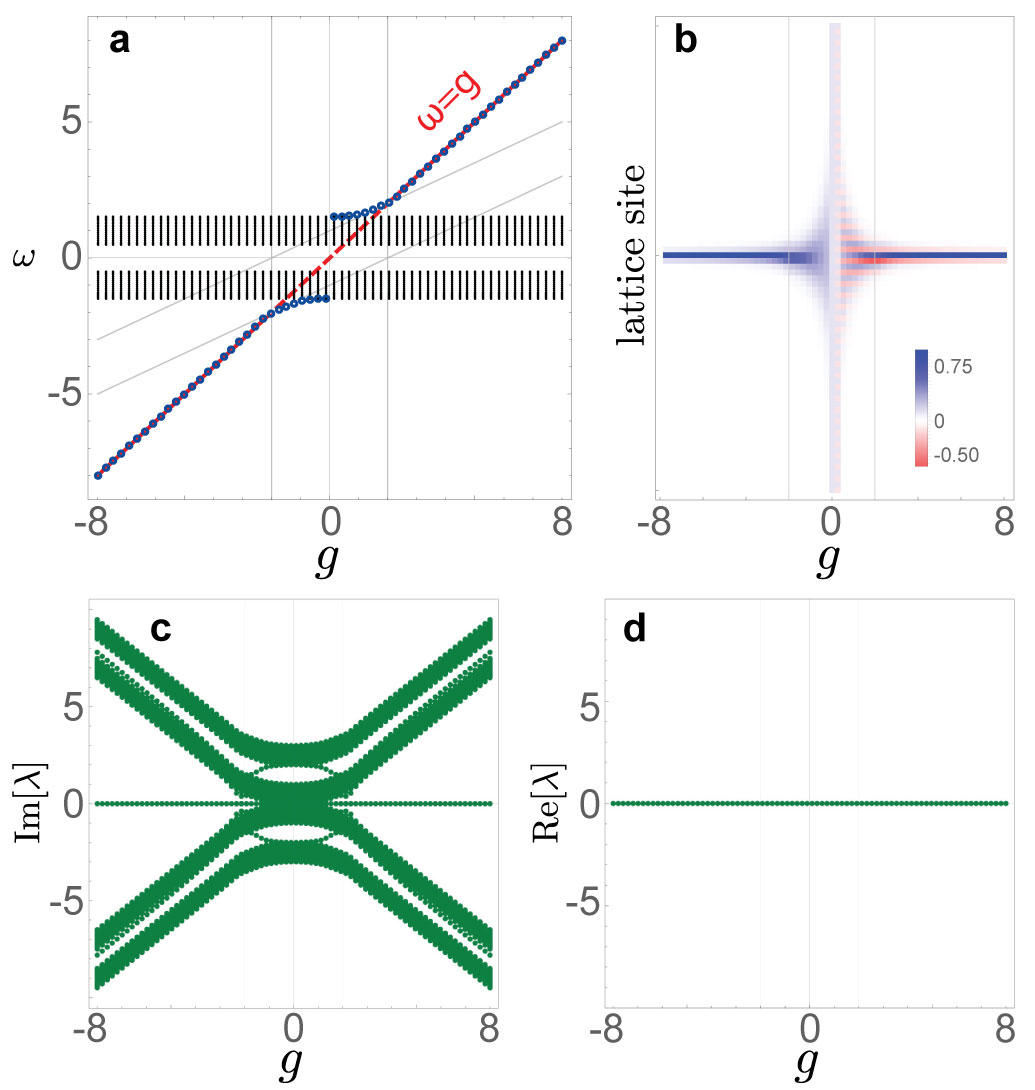}
\caption{Numerical breather solutions of Type ${(\textrm{I}, \pm)}$ ($|g| > 2$) and Type ${(\textrm{II}, \pm)}$ ($|g| \leq 2$) in the nonlinear SSH chain (Eqs.~\eqref{eq:time_indep_DNLSE},~\eqref{eq: SSH}) with $m=1$ and $\epsilon=0.5$. \textbf{a}~Linear spectrum (black) and breather frequencies (blue). \textbf{b}~Spatial profiles $\phi_{(R,\alpha)}$. Blue (red) color represents positive (negative) components. \textbf{c, d} Linear stability eigenvalues $\lambda$.
}
\label{fig:numerics_ssh_type1}
\end{figure}

\subsection{Molecular breather limit} \label{sec: ssh_molecular_limit}
The molecular limit of the SSH model is obtained by setting $\epsilon=0$. We note that there is another limit -- setting $m=0$ -- which is physically equivalent and corresponds to the molecular limit after a change of unit cell. Without loss of generality, we focus here on the limit $\epsilon=0$ and set $m>0$. Our findings are summarized in Fig.~\ref{fig: ssh_theory}.

\subsubsection{Solving for molecular breathers}
In the molecular limit, the discrete breather conditions Eq.~\eqref{eq:time_indep_DNLSE_molecular_limit_reduced} in a single unit cell read
\begin{align}
\omega \phi_A &= g \phi_A^3 - m \phi_B, \\
\omega \phi_B &= -m \phi_A + g \phi_B^3.
\end{align}

To solve these equations, note that adding and subtracting them respectively yields
\begin{align}
    (\phi_A+\phi_B)\left(\phi_A^2-\phi_A\phi_B + \phi_B^2 - \frac{\omega+m}{g}\right)&= 0 ,\\
    (\phi_A-\phi_B)\left(\phi_A^2+\phi_A\phi_B + \phi_B^2 - \frac{\omega-m}{g}\right)&= 0.
\end{align}
If $\phi_A = \pm \phi_B$, we obtain the normalized solutions
\begin{equation}
\begin{pmatrix}\phi_A \\ \phi_B
\end{pmatrix} = \frac{1}{\sqrt{2}} \begin{pmatrix} 1 \\ \pm 1 \end{pmatrix}, \quad \omega = \mp m + \frac{g}{2}.
\end{equation}
If $\phi_A \neq \pm \phi_B$, we obtain the pair of coupled quadratic equations 
\begin{align}
    \phi_A^2-\phi_A\phi_B + \phi_B^2 &= \frac{\omega+m}{g}, \\
    \phi_A^2+\phi_A\phi_B + \phi_B^2 &= \frac{\omega-m}{g}.
\end{align}
Adding these equations and imposing normalization yields
\begin{equation}
\frac{\omega}{g} = \phi_A^2 + \phi_B^2 = 1
\end{equation}
while subtracting them and imposing normalization again, we deduce that 
\begin{equation}
\begin{pmatrix}\phi_A \\ \phi_B
\end{pmatrix} = \begin{pmatrix} x_{\pm} \\ x_{\mp} \end{pmatrix}, \quad \omega = g,
\end{equation}
where is is convenient to define $x_{\pm}= \pm\sqrt{(1 \pm \sqrt{1 - 4m^2/g^2})/2}$ with the key property that $x_{\mp} = -m/(gx_{\pm})$.

We deduce that there are two non-trivial pairs of discrete breather solutions in the molecular limit, which we can summarize as
\begin{equation} \label{eq: ssh_molecular_explicitsoln}
\begin{aligned}
&\phi^{(\mathrm{I},\pm)}_{\alpha} = \begin{pmatrix} x_\pm \\ x_\mp \end{pmatrix}_\alpha, \quad \omega^{(\mathrm{I},\pm)} = g, \\
&\phi^{(\mathrm{II},\pm)}_{\alpha} = \frac{1}{\sqrt{2}} \begin{pmatrix} 1 \\ \pm 1 \end{pmatrix}_\alpha, \quad \omega^{(\mathrm{II},\pm)} = \mp m + \frac{g}{2}.
\end{aligned}
\end{equation}
Since we set $\phi_\alpha$ to be real, these are unique up to a global sign flip.

Solutions $(\mathrm{I},\pm)$ only exist for $|g| \geq 2m$ where the inner square-root in the definition of $x_\pm$ is real. These solutions are single-site localized in the limit of large $g$, and converge to $\phi^{(\mathrm{I},+)}_\alpha \rightarrow (1,0)_\alpha^\mathrm{T}$ and $\phi^{(\mathrm{I},-)}_\alpha \rightarrow (0,1)_\alpha^\mathrm{T}$ as the nonlinearity strength $g\rightarrow \infty$. In the opposite limit $|g| \to 2m^+$ they coincide both with each other and with the appropriate inversion-symmetric solution. Note that for fixed values of $|g| > 2m$ these solutions do not individually preserve inversion symmetry and are not continuously connected to any Wannier state. Instead, they form a frequency-degenerate doublet of partially single-site localized solutions.

Solutions $(\mathrm{II},\pm)$, which exist for all $g \neq 0$, describe an inversion-even (for $+$) or inversion-odd (for $-$) two-site breather mode. As $g \to 0$, these modes merge with the linear band at frequency $\omega = \mp m$, which has inversion-even or inversion-odd maximally-localized Wannier states respectively, with Wannier centers that are equal to the breather centers of mass. Correspondingly, the small-$g$ breathers of the nonlinear SSH chain indeed recover the Wannier states of the band whence they bifurcate with respect to their symmetry data and center of mass, as was also found by Refs.~\onlinecite{ablowitz2021peierls,jurgensen2021,jurgensen2021chern,mostaan2022quantized,jurgensen2022quantized}.

\begin{figure}[t]
\centering
\includegraphics[width=0.48\textwidth,page=1]{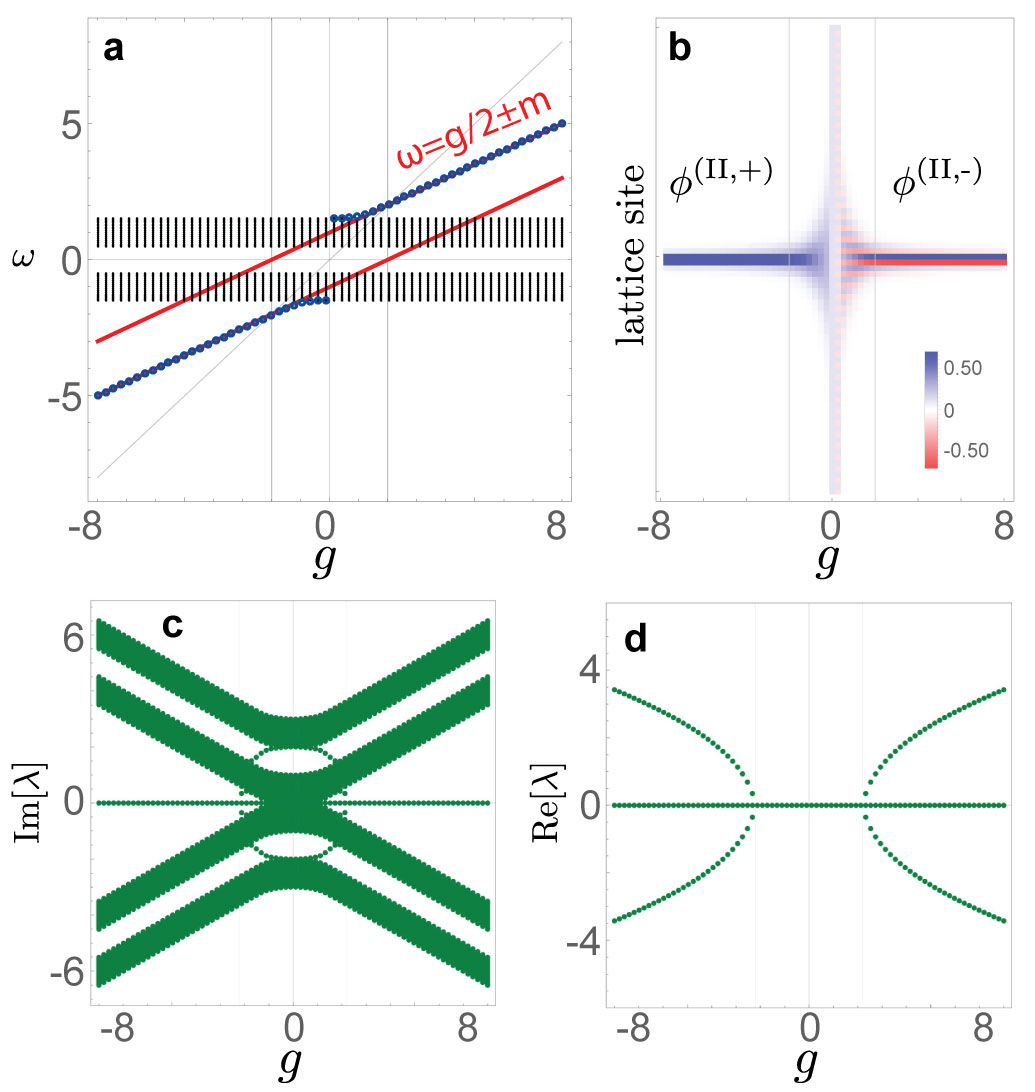}
\caption{Numerical breather solutions of Type $(\textrm{II},\pm)$ in the nonlinear SSH chain (Eqs.~\eqref{eq:time_indep_DNLSE},~\eqref{eq: SSH}) with $m=1$ and $\epsilon=0.5$. \textbf{a}~Linear spectrum (black) and breather frequencies (blue). \textbf{b}~Spatial profiles $\phi_{(R,\alpha)}$. Blue (red) color represents positive (negative) components. \textbf{c, d}~Stability eigenvalues $\lambda$.}
\label{fig:numerics_ssh_type2}
\end{figure}

Each of the above solutions and their stability spectra are depicted in Fig. \ref{fig: ssh_theory}. (See Appendix \ref{app:SSHstab} for exact expressions for these stability spectra.)

Finally, it is interesting to note that the dynamics of the nonlinear SSH chain in the molecular limit defines an integrable dynamical system since, in this limit, there are four real degrees of freedom and two conserved charges (corresponding to energy and particle number) in each unit cell. Constructing action-angle variables for this system might allow for a mathematical proof of the existence~\cite{MacKay1994,Aubry,johansson1997existence} of discrete breathers in some neighbourhood of $\epsilon=0$. We will not pursue this direction here and instead rely on numerical evidence for the local existence of breathers.

\begin{figure*}[t]
\centering
\includegraphics[width=.8\textwidth,page=1]{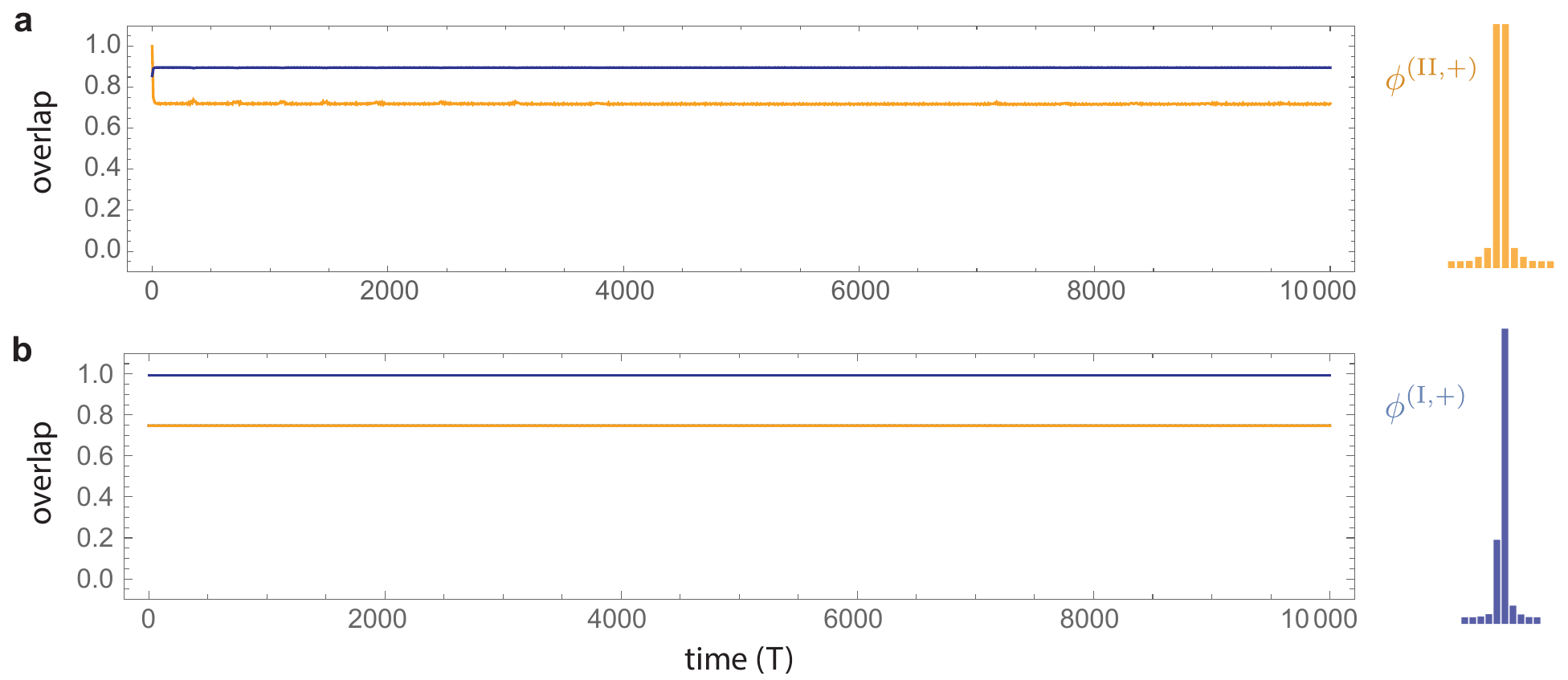}
\caption{Time evolution of linearly stable versus linearly unstable breathers in the nonlinear SSH chain under small perturbations. The parameters of the SSH chain are $m=1$ and $\epsilon=0.15$. The strength of nonlinearity is set to $g=-4$. At this value of $g$, there are two breather solutions: $\phi^{(\mathrm{I},+)}$, stable, and $\phi^{(\mathrm{II},+)}$, unstable. In (a), the initial field is $\phi^{(\mathrm{a})}(t=0)=\phi^{(\mathrm{II},+)}+\phi_\mathrm{pert}$, while in (b), it is $\phi^{(\mathrm{b})}(t=0)=\phi^{(\mathrm{I},+)}+\phi_\mathrm{pert}$, where $\phi_\mathrm{pert}$ is a random field drawn from a uniform distribution between $\pm 0.02$. The fields $\phi^{(\mathrm{a})}(t=0)$ and $\phi^{(\mathrm{b})}(t=0)$ are normalized to have unit magnitude, which is preserved throughout the time evolution. The yellow and blue lines correspond to the overlaps of the instantaneous fields with the unperturbed breathers $\phi^{(\mathrm{II},+)}$ and $\phi^{(\mathrm{I},+)}$, respectively. The simulations are performed using the RK4 algorithm with time step $\Delta t = 0.005$ and run for 10,000 breather periods.
}
\label{fig:numerics_ssh_Dynamics}
\end{figure*}

\begin{figure}[t]
\centering
\includegraphics[width=0.48\textwidth,page=1]{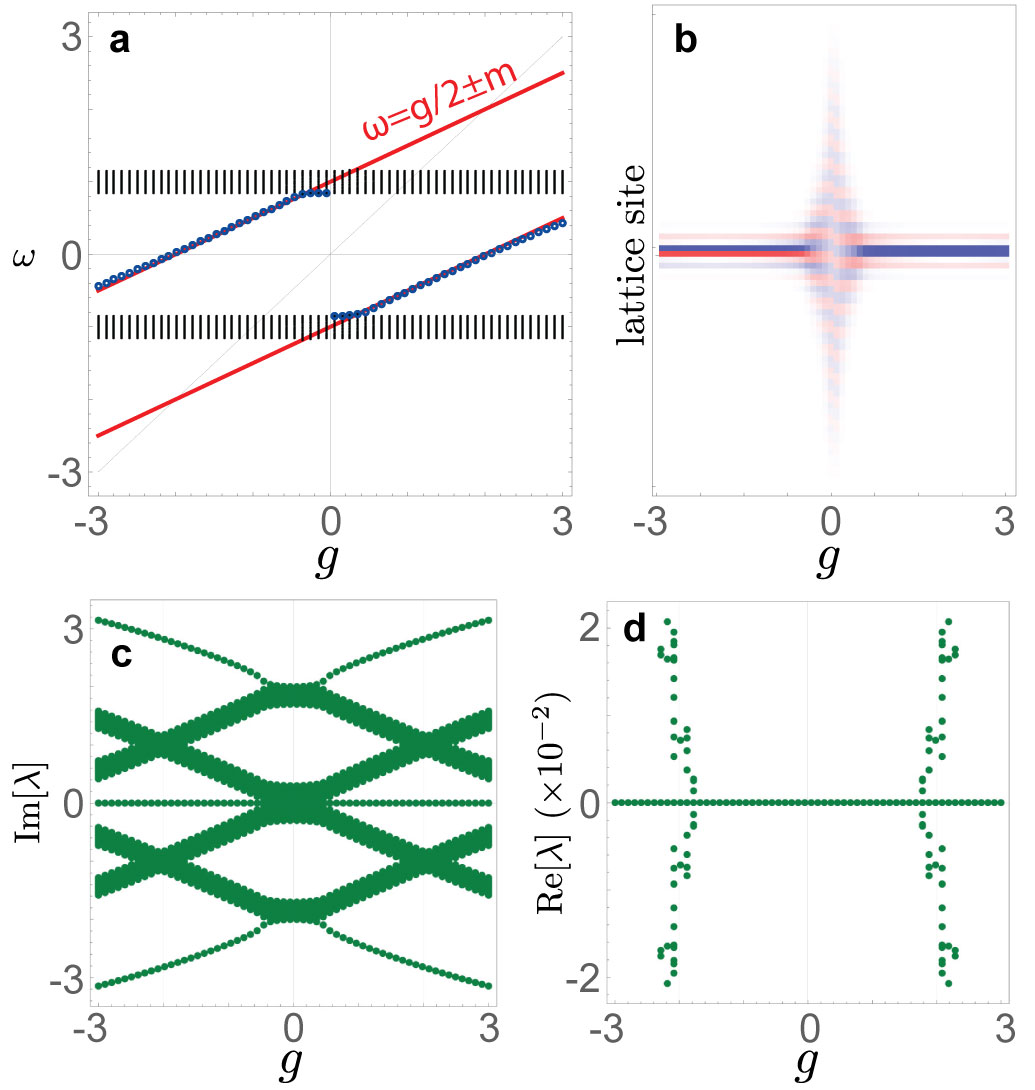}
\caption{Numerical in-gap breather solutions of Type $\mathrm{(II,\pm)}$ in the nonlinear SSH chain (Eqs.~\eqref{eq:time_indep_DNLSE},~\eqref{eq: SSH}) with $m=1$ and $\epsilon=0.15$. \textbf{a}~Linear spectrum (black) and breather frequencies (blue). \textbf{b}~Spatial profiles $\phi_{(R,\alpha)}$. Blue (red) color represents positive (negative) components. \textbf{c, d} Stability eigenvalues $\lambda$. Instabilities develop in the region $|g| \approx 2m$ (absent in the molecular limit), where stability eigenvalues with opposite Krein eigenvalues collide (Fig.~\ref{fig: ssh_theory}c,d).
}
\label{fig:numerics_ssh_InGap}
\end{figure}

\subsection{Perturbing away from the molecular limit} \label{sec: ssh_numerics}
We now solve numerically for discrete breathers in the nonlinear SSH chain away from the molecular limit, where the linear band structure acquires a dispersion [Eq.~\eqref{eq: linearblochssh}] and analytical tractability is no longer expected. The results are shown in Figs.~\ref{fig:numerics_ssh_type1},~\ref{fig:numerics_ssh_type2}  and~\ref{fig:numerics_ssh_InGap}. Let us first describe the  $\phi^\mathrm{(I,R)}$ breathers. These bifurcate from band edges with Bloch states at $k=0$ and with inversion eigenvalues $\lambda^{(\mathcal{I})}=-1$ ($\lambda^{(\mathcal{I})}=+1$) for $g>0$ ($g<0$), see Fig.~\ref{fig:numerics_ssh_type1}a. Consequently, these breathers have components out of phase (in phase) for $g>0$ ($g<0$), see Fig.~\ref{fig:numerics_ssh_type1}b. As expected from the stability eigenvalues of these breathers in the molecular limit (Fig.~\ref{fig: ssh_theory}b), they are always linearly stable (Figs.~\ref{fig:numerics_ssh_type1}c,d).

Breathers $\phi^\mathrm{(II,R)}$ have molecular limit solutions that traverse the linear spectrum. We will analyze these breathers in two regions: when they bifurcate away from the linear bands as $\omega=g/2\pm m$ for $g>0$ ($g<0$) respectively (Fig.~\ref{fig:numerics_ssh_type2}) and when they lie within the linear band gap (Fig.~\ref{fig:numerics_ssh_InGap}). 
In Figs.~\ref{fig:numerics_ssh_type2}b,c, these breathers are qualitatively similar to the breather excitations of the usual DNLSE (Figs.~\ref{fig:numerics_metal}c,d in the appendix); however, unlike for the DNLSE, there is a sharp transition to instability at $g=2$ (Fig.~\ref{fig:numerics_ssh_type2}d), where a pair of stability eigenvalues merge at $\lambda=0$ with equal vanishing Krein eigenvalues $\gamma=0$ (small dots in Fig.~\ref{fig: ssh_theory}c,d).

To probe these breathers' nonlinear stability beyond their linear stability discussed above, we simulate their time evolution under small initial perturbations. This evolution is simulated using the standard fourth-order Runge-Kutta algorithm (RK4). Figure \ref{fig:numerics_ssh_Dynamics} shows the dynamics of breathers $\phi^{(\mathrm{II},+)}+\phi_\mathrm{pert}$ (a) and $\phi^{(\mathrm{I},+)}+\phi_\mathrm{pert}$ (b). In both cases, the breathers are normalized to have unit total charge, and $\phi_\mathrm{pert}$ is a random field drawn from a uniform distribution between $\pm 0.02$. The nonlinearity is set to $g=-4$. At this strength of nonlinearity, $\phi^{(\mathrm{I},+)}$ ($\phi^{(\mathrm{II},+)}$) is a stable (unstable) breather solution. Figure~\ref{fig:numerics_ssh_Dynamics}(a) shows how the unstable breather $\phi^{(\mathrm{II},+)}$ decays into the stable breather $\phi^{(\mathrm{I},+)}$ on the accessible timescale, which can be seen by computing its overlap with the unperturbed breathers. In Fig.~\ref{fig:numerics_ssh_Dynamics}(b), on the other hand, we verify that the shape of the breather $\phi^{(\mathrm{I},+)}$ is preserved under time evolution. The breather periods are $420$ and $320$ time steps for $\phi^{(\mathrm{II},+)}$ and $\phi^{(\mathrm{I},+)}$. Fig.~\ref{fig:numerics_ssh_Dynamics} shows that this behaviour is qualitatively unchanged over $10^4$ breather periods.


The $\phi^\mathrm{(II,R)}$ in-gap breathers bifurcate from Bloch states at $k=\pi$ and inversion eigenvalues $\lambda^{(\mathcal{I})}=-1$ ($\lambda^{(\mathcal{I})}=+1$) for $g<0$ ($g>0$), see Figs.~\ref{fig:numerics_ssh_InGap}a,b. At $g=\pm 2$, the breathers are at mid-gap, \emph{i.e.}, they have $\omega=0$. Exactly at this point in the molecular limit, there is a crossing of real $\lambda$ eigenvalues with opposite Krein eigenvalues (condition Eq.~\eqref{eq: mysterycondition} is fulfilled), allowing in principle for perturbative linear instability (see Figs.~\ref{fig: ssh_theory}c,d). Away from the molecular limit, $\epsilon \neq 0$, the region of linear instability spreads to the vicinity of $g=\pm 2$ (Figs.~\ref{fig:numerics_ssh_InGap}c,d). 

Interestingly, the mid-gap breathers resemble zero energy modes in gapped lattices with chiral symmetry insofar as they have support over only one sublattice ($\alpha = A$ or $\alpha = B$ in Eq.~\eqref{eq: SSH}) on either side of the breather away from its center of mass. At the breather center itself, the breather is supported over two adjacent sites belonging to different sublattices, creating a domain wall where the sublattice supporting the breather mode switches between $A$ and $B$ (Fig.~\ref{fig:numerics_ssh_InGap}b).

\subsection{Comparison of breathers and Wannier states}
We now compare the breather solutions we have predicted analytically and observed numerically with the Wannier states of the linear band structure.

Let us first recall the Wannier states of the SSH chain, with linear Bloch matrix given in Eq.~\eqref{eq: linearblochssh}. In the molecular limit, $\epsilon = 0$, there are two flat bands at frequencies $\omega = \pm m$. In this limit, the Wannier states become eigenstates of the real space hopping matrix Eq.~\eqref{eq: SSH}~\cite{Read17,schindler21noncompact} with eigenfrequencies $\omega = \pm m$, and in the unit cell at $R = 0$ they read 
\begin{equation} \label{eq: ssh_wannierstates}
    W^{\pm m}_{(R,\alpha)} = \delta_{R,0} W^{\pm m}_{\alpha}, \quad W^{\pm m}_{\alpha} = \frac{1}{\sqrt{2}} \begin{pmatrix} 1 \\ \mp 1 \end{pmatrix}_{\alpha}.
\end{equation}
These Wannier states $W^{\pm m}$ have inversion eigenvalues $\lambda_{\mathcal{I}} = \mp 1$ and center of mass $R = 0$. Neither their inversion eigenvalue nor their center of mass can change without breaking inversion symmetry in Eq.~\eqref{eq: linearblochssh} or closing a gap between the two bands: both are topological invariants of the linear band structure. 

We now compare the Wannier states of Eq.~\eqref{eq: ssh_wannierstates} with the breathers of the nonlinear SSH model. As derived in Sec.~\ref{sec: ssh_molecular_limit}, the breather solutions $(\mathrm{I},\pm)$ approach the fully on-site breathers of the DNLSE (see also App.~\ref{appendix: dnlse}) in the limit $g \rightarrow \infty$. Conversely, since these breathers only exist for relatively large nonlinearity strengths $|g| \geq 2m$, there is no sense in which these breathers connect continuously to the band structure as $g \rightarrow 0$ and they are not expected to conform with any Wannier states. Thus we expect that as $g$ becomes much larger than the band gap (and, in general, the bandwidth), $|g| \gg 2m$, the effect of the band structure on the breather spectrum can be neglected aside from possible implications for the breathers' stability.

\begin{figure}[t]
\centering
\includegraphics[width=0.3\textwidth,page=3]{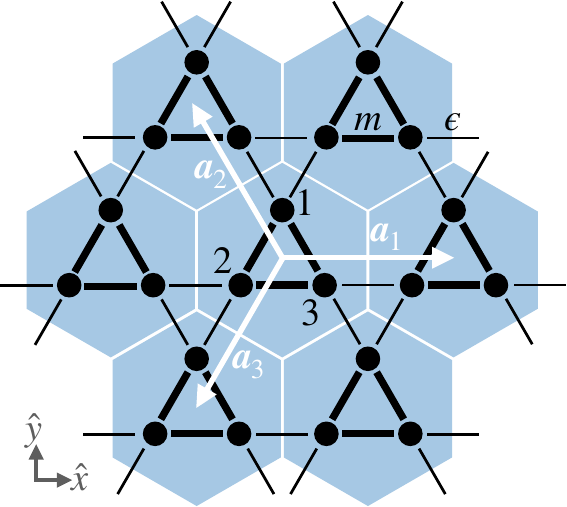}
\caption{Lattice vectors (white arrows), unit cell (blue hexagon), lattice sites (black dots), and hoppings (black lines) for the breathing kagome lattice model in Eq.~\eqref{eq: c3model_realspacehamiltonian}. Strong hoppings inside the unit cell (weighted by $m$) are indicated by thick lines while weak hoppings (weighted by $\epsilon$), which connect different unit cells, are indicated by thin lines.}
\label{fig: c3_model}
\end{figure}

On the other hand, solutions $(\mathrm{II},\pm)$ merge with the linear bands at frequencies $\omega = \mp m$ as $g \rightarrow 0$. We see from comparing Eqs.~\eqref{eq: ssh_molecular_explicitsoln} and~\eqref{eq: ssh_wannierstates} that the breathers $(\mathrm{II},\pm)$ are identical with the Wannier states $W^{\mp m}$.

We deduce that breathers and Wannier states agree in the nonlinear SSH model for sufficiently small nonlinearity strengths $g$. Since two fixed-time slices of a Thouless pumping cycle are topologically equivalent to the SSH model, this result is compatible with earlier findings in Refs.~\onlinecite{ablowitz2021peierls,jurgensen2021,jurgensen2021chern,mostaan2022quantized,jurgensen2022quantized}, which studied nonlinear Thouless pumps and found that the breather center adiabatically follows the Wannier center. In addition to confirming this result in the special case of an inversion-symmetric SSH chain, we have here shown that it extends also to the symmetry data (the inversion eigenvalues) of the breather modes.

\section{Nonlinear breathing kagome lattice} \label{sec: kagome}
We close with an analysis of breather modes resulting from on-site nonlinearity in the 2D breathing kagome lattice with $\mathcal{C}_3$ rotational symmetry.

\subsection{Tight-binding matrix and symmetries} \label{sec: c3modelsymmetry}
We consider a hexagonal crystal with lattice vectors $\bs{a}_1 = (1,0)^\mathrm{T}$ and $\bs{a}_2 = (-1/2,\sqrt{3}/2)^\mathrm{T}$. We also define $\bs{a}_3 = -\bs{a}_1 - \bs{a}_2$ for convenience (see Fig.~\ref{fig: c3_model}). The real-space hopping matrix then reads 
\begin{widetext}
\begin{equation} \label{eq: c3model_realspacehamiltonian}
h_{(\bs{R},\alpha),(\bs{R}',\beta)} =
\begin{pmatrix}
0 & -m \delta_{\bs{R},\bs{R}'} - \epsilon \delta_{\bs{R}-\bs{a}_3,\bs{R}'} & -m \delta_{\bs{R},\bs{R}'} - \epsilon \delta_{\bs{R}+\bs{a}_2,\bs{R}'} \\ 
-m \delta_{\bs{R},\bs{R}'} - \epsilon \delta_{\bs{R}+\bs{a}_3,\bs{R}'} & 0 & -m \delta_{\bs{R},\bs{R}'} - \epsilon \delta_{\bs{R}-\bs{a}_1,\bs{R}'} \\ 
-m \delta_{\bs{R},\bs{R}'} - \epsilon \delta_{\bs{R}-\bs{a}_2,\bs{R}'} &
-m \delta_{\bs{R},\bs{R}'} - \epsilon \delta_{\bs{R}+\bs{a}_1,\bs{R}'} & 0
\end{pmatrix}_{\alpha \beta},
\end{equation}
where $m$ and $\epsilon$ are intra- and inter-unit cell hopping parameters, respectively. After a Fourier transform, the corresponding momentum-space Bloch matrix reads
\begin{equation} \label{eq: kagome_bloch}
    h(\bs{k})_{\alpha \beta} = 
    \begin{pmatrix}
    0 & -m-\epsilon e^{-i (k_1+k_2)} & -m -\epsilon e^{-i k_2} \\ 
    -m-\epsilon e^{i (k_1+k_2)} & 0 & -m -\epsilon e^{i k_1} \\ 
    -m -\epsilon e^{i k_2} & -m -\epsilon e^{-i k_1} & 0
    \end{pmatrix}_{\alpha \beta},
\end{equation}
where we have defined $k_{1,2} = \bs{k} \cdot \bs{a}_{1,2}$.
\end{widetext}
For $\epsilon=0$ ($m=0$), the spectrum consists of three exactly flat bands at frequencies $\omega \in \{-2,1,1\}m$ ($\omega \in \{-2,1,1\}\epsilon$), respectively, which decompose into a single gapped and two exactly degenerate flat bands. This gap persists away from the flat band limit as long as $|m|\neq|\epsilon|$. We will now study breather modes in the nonlinear problem obtained by combining the hopping matrix Eq.~\eqref{eq: c3model_realspacehamiltonian} with the time-independent nonlinear Schrödinger equation Eq.~\eqref{eq:time_indep_DNLSE}.

Together with lattice translation symmetry, this model features two point group generators: a $\mathcal{C}_3$ symmetry perpendicular to the 2D plane with 
\begin{equation} \label{eq: c3_sym_kagome}
    (\mathcal{C}_3)_{(\bs{R},\alpha),(\bs{R}',\beta)} = \delta_{C_3 \bs{R}',\bs{R}} \begin{pmatrix} 0 & 0 & 1 \\ 1 & 0 & 0 \\ 0 & 1 & 0 \end{pmatrix}_{\alpha \beta},
\end{equation}
where $C_3$ implements a $120^\circ$ rotation perpendicular to the 2D real space plane, and an in-plane $\mathcal{C}_2$ rotation that reads
\begin{equation} \label{eq: c2_sym_kagome}
    (\mathcal{C}_2)_{(\bs{R},\alpha),(\bs{R}',\beta)} = \delta_{C_2 \bs{R}',\bs{R}} \begin{pmatrix} 1 & 0 & 0 \\ 0 & 0 & 1 \\ 0 & 1 & 0 \end{pmatrix}_{\alpha \beta},
\end{equation}
where $C_2$ implements a $180^\circ$ rotation about the $\hat{y}$ axis within the 2D plane. These symmetries satisfy $(\mathcal{C}_3)^3 = (\mathcal{C}_2)^2 = \mathbb{1}$ and generate the dihedral point group symmetry $\mathcal{D}_3$ (Schönflies notation), which is also a symmetry of the full nonlinear problem Eq.~\eqref{eq:time_indep_DNLSE} as confirmed in App.~\ref{appendix: symmetries}. Correspondingly, there are four distinct possibilities for the breather symmetry data: Within a (suitably chosen) unit cell and up to a global sign flip which we remove by fixing the convention that $a > 0$, the breather mode has one of four symmetry types:
\begin{enumerate}[(I)]
\item a sixfold multiplet of the form $\phi_\alpha = (a,b,c)_\alpha^\mathrm{T}$, where $a,b,c$ are all distinct and non-zero,
\item a triplet $\phi_\alpha = (a,b,b)_\alpha^\mathrm{T}$ with $\mathcal{C}_2$-eigenvalue $\lambda^{(2)} = +1$,
\item a triplet $\phi_\alpha = (a,-a,0)_\alpha^\mathrm{T}$ with $\mathcal{C}_2$-eigenvalue $\lambda^{(2)} = -1$,
\item a singlet $\phi_\alpha = (a,a,a)_\alpha^\mathrm{T}$ with $\mathcal{C}_3$-eigenvalue $\lambda^{(3)} = +1$ and $\mathcal{C}_2$-eigenvalue $\lambda^{(2)} = +1$.
\end{enumerate}

\subsection{Molecular breather limit} \label{sec: c3_molecular_limit}
The molecular limit is obtained by setting $\epsilon = 0$. As for the SSH model, there is an alternative limit $m = 0$ that is equivalent up to a change of unit cell. We summarize our findings in Fig.~\ref{fig: oai_theory}.

\subsubsection{Solving for molecular breathers}
In the molecular limit $\epsilon = 0$, the discrete breather conditions Eq.~\eqref{eq:time_indep_DNLSE_molecular_limit_reduced} become
\begin{align}
\label{eq:C3db1}
\omega \phi_A &= g\phi_A^3 - m \phi_B - m\phi_C, \\
\omega \phi_B &= - m \phi_A + g\phi_B^3  - m\phi_C, \\
\omega \phi_C &= - m \phi_A - m\phi_B + g\phi_C^3.
\end{align}
We now consider each possibility for the breather symmetry type in turn. To this end, first note that subtracting the above equations pairwise yields
\begin{align}
\label{eq:C3diff1}
(\phi_A-\phi_B)(\omega - m - g(\phi_A^2+\phi_A\phi_B+\phi_B^2)) &= 0, \\
\label{eq:C3diff2}
(\phi_B-\phi_C)(\omega - m - g(\phi_B^2+\phi_B\phi_C+\phi_C^2)) &= 0, \\
\label{eq:C3diff3}
(\phi_C-\phi_A)(\omega - m - g(\phi_C^2+\phi_C\phi_A+\phi_A^2)) &= 0,
\end{align}
while adding them yields
\begin{equation}
\label{eq:C3sum}
g(\phi_A^3 + \phi_B^3 + \phi_C^3)  - (\omega+2m)(\phi_A+\phi_B+\phi_C)= 0.
\end{equation}
We first show that there are no solutions of symmetry Type I identified above. In this case, $ \phi_A = a > 0$, $\phi_B=b$ and $\phi_C = c$ being pairwise distinct and the normalization condition $a^2 + b^2 +c^2 = 1$ in Eqs. \eqref{eq:C3diff1}-\eqref{eq:C3diff3} imply that
\begin{equation}
ab-c^2 = bc -a^2 = ca - b^2 = \frac{\omega - m}{g} - 1.
\end{equation}
Subtracting any two of these equations pairwise yields the condition 
\begin{equation}
\label{eq:C3zerosum}
a+b+c = 0,
\end{equation}
which in Eq. \eqref{eq:C3sum} yields
\begin{equation}
a^3 + b^3 + c^3 = 0.
\end{equation}
Substituting $c = -(b+a)$ from Eq. \eqref{eq:C3zerosum} into this equation, we deduce that
\begin{equation}
3ab(a+b) = 0,
\end{equation}
which implies that $b=0$ or $c=0$, so that by Eq. \eqref{eq:C3zerosum} the breather has symmetry Type III. We deduce that no molecular breathers of symmetry Type I exist.

\begin{figure}[t]
\centering
\includegraphics[width=0.48\textwidth,page=2]{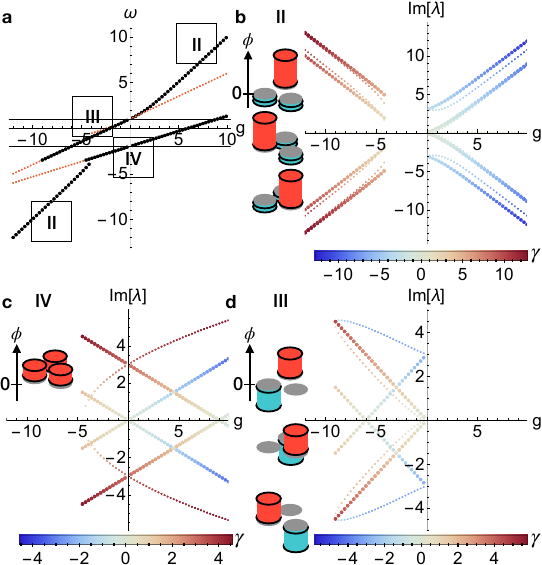}
\caption{Analytical molecular limit breather solutions in the nonlinear breathing kagome lattice (Eqs.~\eqref{eq:time_indep_DNLSE},~\eqref{eq: c3model_realspacehamiltonian}) for $m = 1$, $\epsilon = 0$. \textbf{a}~Breather frequencies derived in Sec.~\ref{sec: c3_molecular_limit}. The frequencies of linearly stable (unstable) breather solutions are shown with thick black (thin red) dots. The linear flat-band spectrum at $\omega \in \{-2,1,1\}$ is indicated by gray horizontal lines. 
At $g\lesssim-4$ and $g>0$, there exist three stable, mostly on-site breather modes with $\mathcal{C}_2$ eigenvalue $\lambda^{(2)}=+1$ and exchanged by $\mathcal{C}_3$ symmetry (Type II). At fixed $g$, their stability eigenvalues shown in panel \textbf{b} all belong to Krein eigenvalues $\gamma$ of the same sign, indicating stability away from the molecular limit. Here, non-degenerate (macroscopically degenerate) stability eigenstates are indicated by thin (thick) dots (see Sec.~\ref{subsec: molecular_stability} for details), and the curves terminate when the respective solution becomes unstable. At $g \gtrsim -4.5$, one further stable breather with $\mathcal{C}_2$ and $\mathcal{C}_3$ eigenvalues $\lambda^{(2)}=\lambda^{(3)}=+1$ (Type IV, panel \textbf{c}) exists that continuously connects with the linear band at $\omega = -2$. For $-9.07m \lesssim g \leq 0$, three more stable two-site breathers with $\mathcal{C}_2$ eigenvalue $\lambda^{(2)}=-1$ and exchanged by $\mathcal{C}_3$ symmetry (Type III, panel \textbf{d}) exist within the gap that connect with the two degenerate bands at $\omega = 1$.}
\label{fig: oai_theory}
\end{figure}

We next turn to molecular breathers of symmetry Type II, with $\phi_A = a$ and $\phi_B = \phi_C = b$. In this case, Eq. \eqref{eq:C3db1} uniquely determines $b$ as a function of $a$ and $\omega$, via
\begin{equation}
b = \frac{a}{2m}(ga^2-\omega).
\end{equation}
However, we can also solve Eq. \eqref{eq:C3diff1} for $b$ to yield
\begin{equation}
\label{eq:quadb}
b = - \frac{a}{2} \pm \sqrt{\frac{\omega-m}{g} - \frac{3a^2}{4}.}
\end{equation}
Solving this pair of equations for $\omega$ we find two distinct solutions as a function of $a$, namely
\begin{equation}
\label{eq:C3TypeIIomega}
\omega_{\pm}(a) = m + ga^2 + \frac{2m^2}{ga^2} \pm m \sqrt{1+ \frac{4m^2}{g^2a^4}},
\end{equation}
with
\begin{equation}
\label{eq:C3TypeIIb}
b_{\pm}(a) = -\left(\frac{a}{2} + \frac{m}{ga} \pm \sqrt{\left(\frac{a}{2}\right)^2+\left(\frac{m}{ga}\right)^2}\right).
\end{equation}
It remains to solve for $a$, which is constrained by the normalization condition $a^2+2b(a)^2 = 1$. The latter implies a cubic equation for $a^2$ of the form
\begin{equation}
\label{eq:C3TypeIIa}
3a^6 + \left(\frac{4m}{g}-4\right)a^4 + \left(1+ \frac{12m^2}{g^2}-\frac{4m}{g}\right)a^2 - \frac{8m^2}{g^2} = 0.
\end{equation}
To analyse these equations it is simplest to first consider the limit of large nonlinearity $|g| \to \infty$. In this limit, Eqs. \eqref{eq:C3TypeIIa} and \eqref{eq:C3TypeIIb} yield
\begin{equation}
a^2(3a^2-1)(a^2-1) = 0, \quad b_{\pm}(a) = - \left(\frac{a}{2} \pm\frac{a}{2}\right).
\end{equation}
Imposing the choice $a>0$, we deduce that there are two distinct branches of solutions as $|g| \to \infty$, each of which picks out a different choice of sign in Eq. \eqref{eq:C3TypeIIomega}, namely
\begin{align}
\phi^{\mathrm{II}}_+ \sim  \frac{1}{\sqrt{3}}\begin{pmatrix} 1 \\ -1 \\ -1\end{pmatrix}, \quad \omega^{\mathrm{II}}_+ \sim \frac{g}{3}, \quad |g| \to \infty
\end{align}
and
\begin{align}
\phi^{\mathrm{II}}_- \sim \begin{pmatrix} 1 \\ 0 \\ 0\end{pmatrix}, \quad  \omega^{\mathrm{II}}_- \sim g, \quad |g| \to \infty.
\end{align}
In the opposite limit $|g| \to 0^+$, we find that
\begin{equation}
3a^2-2 = 0, \quad b_{\pm}(a) = - \left(\frac{a}{2} + \frac{m}{ga} \pm \frac{m}{ga} + \mathcal{O}(g/m) \right),
\end{equation}
which picks out the normalizable solution
\begin{equation} \label{eq: kagome_breather_ii_gtozero}
\phi^{\mathrm{II}}_- \sim \frac{1}{\sqrt{6}}\begin{pmatrix} 2 \\ -1 \\ -1 \end{pmatrix}, \quad  \omega^{\mathrm{II}}_- \sim m, \quad |g| \to 0.
\end{equation}
Thus only the $\phi^{\mathrm{II}}_-$ solution, with
\begin{equation}
\phi^{\mathrm{II}}_- = \begin{pmatrix} a \\ b_{-}(a) \\ b_-(a) \end{pmatrix}, \quad \omega^{\mathrm{II}}_- = \omega_-(a), 
\end{equation}
with $\omega_-(a)$ and $b_-(a)$ given by Eqs. \eqref{eq:C3TypeIIomega} and \eqref{eq:C3TypeIIb} respectively, and $a$ fixed by normalization $a^2 + 2b_-(a)^2 = 1$, is able to interpolate between $|g| = \infty$ and $|g| = 0^+$ and therefore connect continuously to the linear band.

By solving the system numerically, we find that these two limits are in fact continuously connected such that the solution $\phi^{\mathrm{II}}_-$ always forms a frequency-degenerate triplet orbit of $\mathcal{C}_3$ with $\mathcal{C}_2$-eigenvalue $\lambda^{(2)} = +1$ (see end of Sec.~\ref{sec: c3modelsymmetry}).

Molecular breathers of the remaining symmetry types are rather easier to determine. The normalized Type III solution with our sign convention $\phi_A = a> 0$ has $a = 1/\sqrt{2}$, and its frequency follows by e.g. Eq. \eqref{eq:C3db1}. Thus the Type III molecular breather is given by
\begin{equation} \label{eq: kagome_iii_solution}
\phi^\mathrm{III} = \frac{1}{\sqrt{2}}\begin{pmatrix}
    1 \\ -1 \\ 0
\end{pmatrix}, \quad
\omega^\mathrm{III} = m + \frac{g}{2},
\end{equation}
and corresponds to a triplet orbit of $\mathcal{C}_3$ with $\mathcal{C}_2$-eigenvalue $\lambda^{(2)} = -1$.

Finally, the normalized Type IV solution with $\phi_A > 0$ is straightforwardly found to satisfy
\begin{equation} \label{eq: c3_symmetric_breathermode}
\phi^\mathrm{IV} = \frac{1}{\sqrt{3}}\begin{pmatrix}
    1 \\ 1 \\ 1
\end{pmatrix}, \quad
\omega^\mathrm{IV} = -2m + \frac{g}{3}.
\end{equation}
This corresponds to a singlet with $\mathcal{C}_3$-eigenvalue $\lambda^{(3)} = +1$ and $\mathcal{C}_2$-eigenvalue $\lambda^{(2)} = +1$.

\begin{figure}[t]
\centering
\includegraphics[width=\columnwidth]{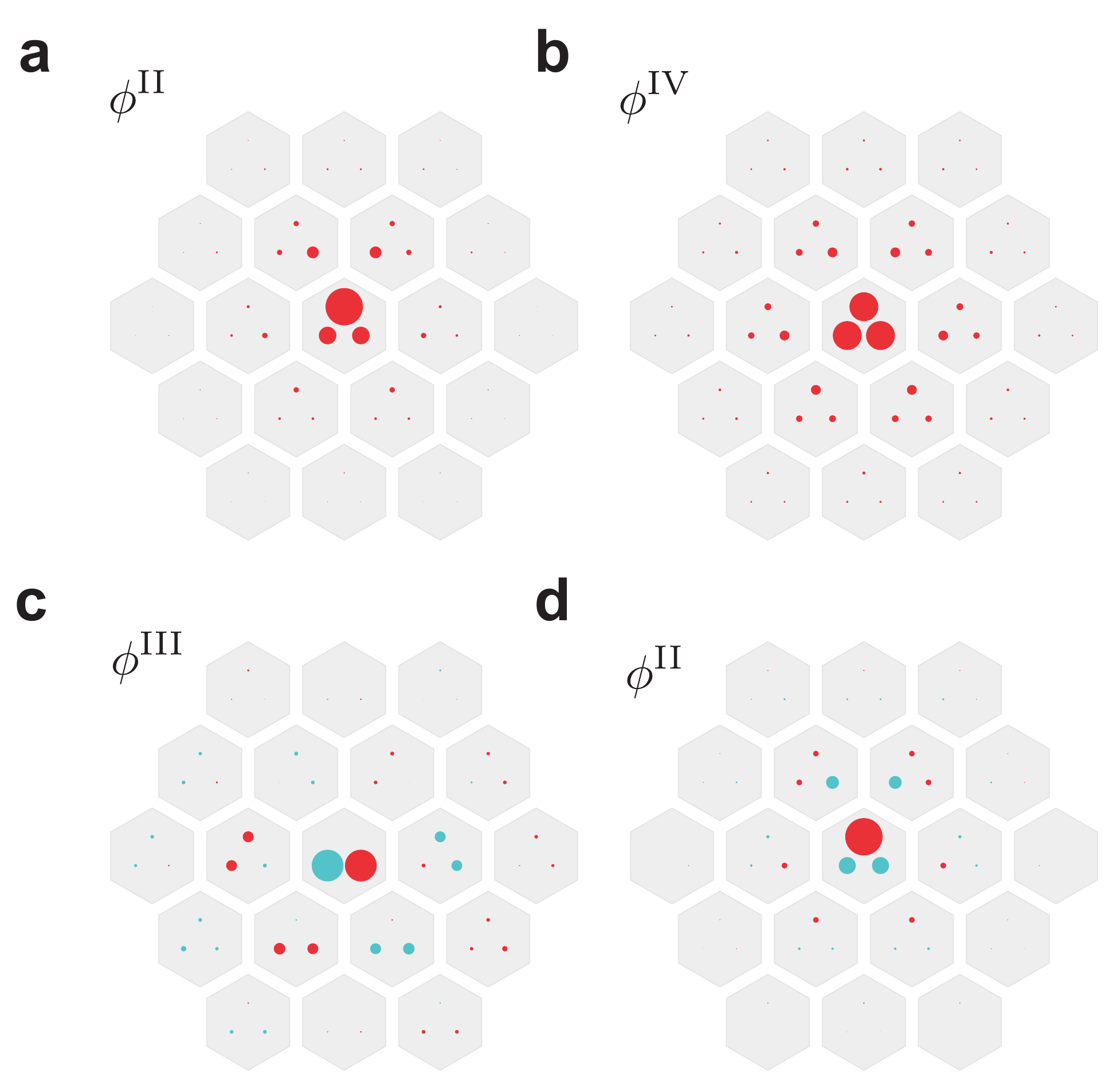}
\caption{Numerical in-gap breather profiles $\phi_{(\bs{R},\alpha)}$ in the nonlinear breathing kagome lattice (Eqs.~\eqref{eq:time_indep_DNLSE},~\eqref{eq: c3model_realspacehamiltonian}) away from the molecular limit. In all plots, $m=1$; in \textbf{a,d} $\epsilon=0.5$, while in \textbf{b,c} $\epsilon=0.25$; additionally, $g=-6$ and $\omega=-5.89$ in \textbf{a}, $g=-4$ and $\omega=-3.34$ in \textbf{b}, $g=-4$ and $\omega=-0.96$ in \textbf{c}, and $g=+4$ and $\omega=4.1$ in \textbf{d}. The area of the disks is proportional to the magnitude of the field and the red (teal) color represents a sign of $+1$ ($-1$).}
\label{fig: c3_numerical_breathers}
\end{figure}

The dispersion of these solutions with $g$, their linear stability eigenvalues (Sec.~\ref{sec:linear_stability}), and their Krein eigenvalues (Sec.~\ref{sec: kreinsignature}) are summarized in Fig.~\ref{fig: oai_theory}. For a detailed results on these stability spectra, including closed form expressions where available, we refer to Appendix \ref{app:BKLstab}.

\begin{figure}[t]
\centering
\includegraphics[width=\columnwidth]{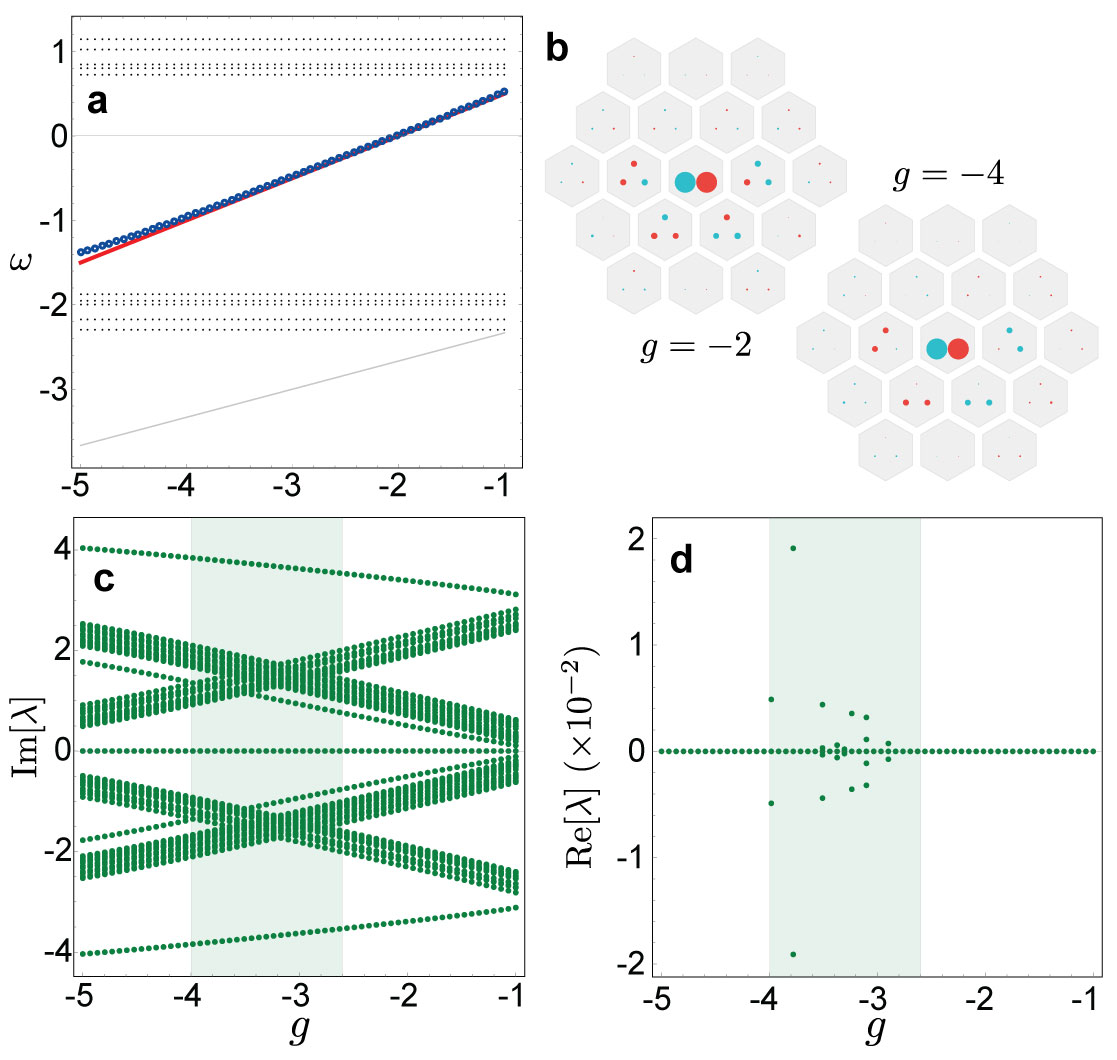}
\caption{Numerical breather solutions of Type III in the nonlinear breathing kagome lattice (Eqs.~\eqref{eq:time_indep_DNLSE},~\eqref{eq: c3model_realspacehamiltonian}) with $m=1$, $\epsilon=0.25$. {\bf a} Analytical (red line) and numerical (blue dots) frequency $\omega$. {\bf b} Breather profiles $\phi_{(\bs{R},\alpha)}$ for $g=-2$ and $g=-4$. {\bf c, d} Stability eigenvalues $\lambda$. Instabilities develop in the region $|g| \approx -3m$ (absent in the molecular limit), where stability eigenvalues with opposite Krein eigenvalues collide (Fig.~\ref{fig: oai_theory}d).}
\label{fig: c3_numerical_breathers_stability}
\end{figure}

\subsection{Perturbing away from the molecular limit} \label{sec: c3numerics}
Away from the molecular limit, i.e., when $\epsilon \neq 0$, the breather solutions we have derived spread beyond their home unit cell (situated at $\bs{R} = \bs{0}$ in Eq.~\eqref{eq: mostlocalized_soln}). Examples of these breathers for all Types II-IV are shown in Fig.~\ref{fig: c3_numerical_breathers} for non-zero $\epsilon$. In all cases, the symmetry data of the molecular breather is preserved. As in the 1D case, linear instabilities away from the molecular limit indeed only arise when stability eigenvalues with opposite Krein eigenvalues collide, as is evident from comparing Fig.~\ref{fig: c3_numerical_breathers_stability} with Fig.~\ref{fig: oai_theory}d. Robustness of our predictions to the full nonlinear dynamics of the nonlinear breathing Kagome lattice Hamiltonian is established through direct numerical simulation in Appendix \ref{app:StabTest2D}.

\subsection{Comparison of breathers and Wannier states}
We now compare the breather solutions we have found analytically and confirmed numerically with the Wannier states of the linearized band structure.

Let us first recall the Wannier states of the linear hopping model, with Bloch matrix given in Eq.~\eqref{eq: kagome_bloch}. In the molecular limit $\epsilon = 0$, there is a single flat band at frequency $\omega = -2m$, and two degenerate flat bands at frequency $\omega = m$. In this limit, the Wannier states become eigenstates of the hopping matrix Eq.~\eqref{eq: c3model_realspacehamiltonian}. 

In the unit cell at $\bs{R} = 0$, the Wannier state at the frequency $\omega = -2m$ reads
\begin{equation} \label{eq: kagome-wanniers}
    W^{-2m}_{(\bs{R},\alpha)} = \delta_{\bs{R},\bs{0}} W^{-2m}_{\alpha}, \quad W^{-2m}_{\alpha} = \frac{1}{\sqrt{3}} \begin{pmatrix} 1 \\ 1 \\ 1 \end{pmatrix}.
\end{equation}
This Wannier state is indeed the same as the breather merging with the band at $\omega = -2m$ as $g \rightarrow 0$: Comparing Eqs.~\eqref{eq: c3_symmetric_breathermode} and~\eqref{eq: kagome-wanniers}, we see that $W^{-2m} = \phi^\mathrm{IV}$.

The Wannier states at the frequency $\omega = m$ are twofold degenerate and so depend on a choice of basis. It is common to choose them to be eigenstates of the $\mathcal{C}_3$ symmetry in Eq.~\eqref{eq: c3_sym_kagome}. However, this convention requires complex amplitudes in the Wannier state wave functions, preventing a meaningful comparison with the breather solutions that are fully real because of TRS. (Note that the breathers do not enjoy the same gauge freedom as the Wannier states, because they solve a nonlinear equation that does not remain invariant under a change of basis.) Therefore, we instead choose the Wannier states to be the following real eigenstates of the $\mathcal{C}_2$ symmetry in Eq.~\eqref{eq: c2_sym_kagome}:
\begin{equation} \label{eq: kagome_wanniers_degenerate}
\begin{aligned}
    &W^{m,1}_{(\bs{R},\alpha)} = \delta_{\bs{R},\bs{0}} W^{m,1}_{\alpha}, \quad W^{m,1}_{\alpha} = \frac{1}{\sqrt{6}} \begin{pmatrix} 2 \\ -1 \\ -1 \end{pmatrix}, \\
    &W^{m,2}_{(\bs{R},\alpha)} = \delta_{\bs{R},\bs{0}} W^{m,2}_{\alpha}, \quad
    W^{m,2}_{\alpha} = \frac{1}{\sqrt{2}} \begin{pmatrix} 0 \\ 1 \\ -1 \end{pmatrix}.
\end{aligned}
\end{equation}
Here, $W^{m,1}$ ($W^{m,2}$) has $\mathcal{C}_2$-eigenvalue $\lambda^{(2)} = 1$ ($\lambda^{(2)} = -1$), respectively. By comparing Eqs.~\eqref{eq: kagome_breather_ii_gtozero},~\eqref{eq: kagome_iii_solution} with Eq.~\eqref{eq: kagome_wanniers_degenerate}, we see that $W^{m,1}_{\alpha} = \phi^{\mathrm{II}}_-$, while $W^{m,2}_{\alpha}$ is related to $\phi^\mathrm{III}$ by a  $\mathcal{C}_3$ rotation. Since $\phi^\mathrm{III}$ is part of a frequency-degenerate triplet of breathers exchanged by $\mathcal{C}_3$ symmetry, both the Wannier states $W^{m,1}$ and $W^{m,2}$ are recovered by breather solutions in the limit $g \rightarrow 0$. 

However, it is important to note that the multiplicities of breathers and Wannier states do not match: For small negative $g < 0$, there are three breather solutions of Type III (Fig.~\ref{fig: oai_theory}) that merge into a single Wannier state $W^{m,2}$ as $g \rightarrow 0$. For small positive $g>0$, the single Wannier state $W^{m,1}$ likewise branches out into three breathers of Type II. This mismatch is imposed by $\mathcal{C}_3$ symmetry: while the two Wannier states $W^{m,1}$ and $W^{m,2}$ span a two-dimensional representation of $\mathcal{C}_3$ symmetry corresponding to the \emph{complex} $\mathcal{C}_3$ eigenvalues $e^{\mathrm{i} \frac{2\pi}{3}}$, the breather solutions must either have real $\mathcal{C}_3$ eigenvalues or form triplets permuted by $\mathcal{C}_3$ (see also Sec.~\ref{sec: symmetry_data_introduction}). It is also interesting to note that different Wannier states, with different $\mathcal{C}_2$ eigenvalues, are recovered by the breather solutions in different one-sided limits $g \rightarrow 0^+$ and $g \rightarrow 0^-$.

\section{Conclusion}
We have introduced a general framework for understanding nonlinear breather excitations and their crystalline symmetry data in band structures with a molecular limit. We then applied this framework to derive the breather spectrum in nonlinear generalizations of the 1D SSH chain and the 2D breathing kagome lattice. In both cases, we found two qualitatively different possible scenarios for nonlinear breathers: (1) on-site or mostly on-site breathers that continuously approximate the delta-localized breathers of the DNLSE in the anti-continous limit where $g \rightarrow \infty$ and break the crystalline symmetry. (2) Wannier-type breathers that are continuously connected to a subset of the Wannier states of the linear band structure in the limit $g \rightarrow 0$ and retain some crystalline symmetry.

Importantly, while we found that the center of mass and crystalline symmetry data of breathers in scenario (2) do in general agree with those of appropriately chosen Wannier states, this does not imply a one-to-one correspondence between breathers and Wannier states. In the case of $N$-fold degenerate bands (or, more generally, $N$ bands that cannot be gapped everywhere), there is a $U(N)$ gauge freedom in defining Wannier states. However, the specific nonlinear terms considered in this paper break $U(N)$ symmetry, so that there is no analogous gauge freedom in defining discrete breathers. Even when we restrict to real Wannier states in systems with TRS, these Wannier states exhibit an $O(N)$ gauge symmetry with no counterpart in the breather spectrum. We note that it is possible to construct nonlinear Schr{\"o}dinger equations that respect additional global symmetries~\cite{fordy1983nonlinear,faddeev2007hamiltonian}, which might lead to a closer correspondence between breathers and Wannier states, but leave more detailed investigation of such nonlinearities to future work.

Thus for nonlinearities of the type studied in this paper, the $g \rightarrow 0$ limit of a given breather solution picks out one allowed Wannier state among possibly infinitely many equally valid choices. We have seen in the example of the nonlinear breathing Kagome lattice that the specific choice of Wannier state directly determines the crystalline symmetry data of the breather solution and thereby also the multiplicity of the associated frequency-degenerate multiplet. Moreover, we have seen that the specific choice of Wannier state -- which fixes the topologically stable properties of the breather mode -- depends sensitively on the direction in which the limit $g \rightarrow 0$ is approached.

In conclusion, we have shown that it is not in general possible to predict the breather spectrum of a given nonlinear lattice model directly from the linear Wannier states, even for band structures for which these can be exponentially localised and made to conform with all crystalline symmetries and TRS. Nevertheless, for band structures with a molecular limit, the formalism developed in this work can be used to analytically predict the existence, linear stability, and crystalline symmetry data of breather solutions. It remains an open question whether a similar analytical understanding could be achieved for topological band structures, which by definition do not admit any molecular limit~\cite{Bradlyn17}.

\begin{acknowledgments}
W.A.B. and F.S. thank the Kavli Institute for Theoretical Physics for hosting during some stages of this work. V.B.B. and F.S. were supported by a fellowship at the Princeton Center for Theoretical Science during some stages of this work. This research was supported in part by the National Science Foundation under Grant No. NSF PHY-1748958. This work was supported by a UKRI Future Leaders Fellowship MR/Y017331/1. W.A.B. thanks the support of startup funds from Emory University.
\end{acknowledgments}

\appendix

\section{Linear stability theory of discrete breathers}
\label{app:stability}
In this Appendix, we provide a self-contained account of the linear stability theory of discrete breathers in tight-binding models with cubic nonlinearity, using the techniques of Hamiltonian stability theory presented in Ref.~\cite{kapitula2013spectral}.
\subsection{Hamiltonian stability theory} \label{sec:linear_stability}
Discrete breather solutions of Eq. \eqref{eq:time_indep_DNLSE} above have the useful property that they are stationary points of a Hamiltonian evolution with respect to $\tilde{H}$. This allows us to study their stability properties using the powerful machinery of Hamiltonian stability theory~\cite{kapitula2013spectral}.

To make the Hamiltonian structure manifest, it is helpful to introduce standard canonical coordinates~\cite{mendl2015low} given by the real and imaginary parts of the field $\widetilde{\Phi}_a$, namely
\begin{equation}
\tilde{\Phi}_a = \frac{1}{\sqrt{2}}(q_a + ip_a),
\end{equation}
with respect to which the time evolution Eq.~\eqref{eq:rotated_frame} can be written as Hamilton's equations
\begin{align}
\dot{p}_a = -\frac{\partial \tilde{H}}{\partial q_a}, \quad
\dot{q}_a = \frac{\partial \tilde{H}}{\partial p_a}.
\end{align}
Introducing the two-by-two symplectic form
\begin{equation}
J = \begin{pmatrix} 0 & -1 \\ 1 & 0 \end{pmatrix},
\end{equation}
we can write Hamilton's equations as a pair of equations
\begin{equation}
\begin{pmatrix} \dot{p}_a \\ \dot{q}_a \end{pmatrix} = J \begin{pmatrix} \partial_{p_a} \tilde{H} \\ \partial_{q_a} \tilde{H} \end{pmatrix}
\end{equation}
for each $a$. Let us now consider perturbing about a stationary point $\dot{\bar{p}}_a= \dot{\bar{q}}_a = 0$, with $p_a(t) = \bar{p}_a + \delta p_a(t)$ and $q_a(t) = \bar{q}_a + \delta q_a(t)$. Then the linearized equations of motion are given by
\begin{equation}
    \begin{pmatrix} \delta \dot{p}_a \\ \delta \dot{q}_a \end{pmatrix} = 
    \sum_{b} J\mathcal{L}_{ab} \begin{pmatrix} \delta p_b  \\ \delta q_b \end{pmatrix},
\end{equation}
where we have defined the matrix
\begin{equation}
\mathcal{L}_{ab} = \begin{pmatrix} \tilde{h}_{ab}  + \frac{g}{2}(\bar{q}_a^2 + 3 \bar{p}^2_a)\delta_{ab} & g \bar{p}_a \bar{q}_a \delta_{ab} \\ g \bar{p}_a \bar{q}_a \delta_{ab}& \tilde{h}_{ab}  + \frac{g}{2}(3\bar{q}_a^2 + \bar{p}^2_a)\delta_{ab} \end{pmatrix}.
\end{equation}
Thus the linear stability of stationary points is determined by the spectrum of the matrix $J\mathcal{L}$, viewed as a $2 N \times 2N$ matrix where $N$ is the dimension of $\tilde{h}$. In particular, any eigenvalues with positive real part will give rise to linear instability. Consider now eigenvalues $\lambda$, such that
\begin{equation} \label{eq: v_definition}
J\mathcal{L} \vec{v} = \lambda \vec{v}
\end{equation}
for some $\vec{v}$. In classical Hamiltonian systems, all eigenvalues $\lambda$ of $J\mathcal{L}$ are constrained~\cite{kapitula2013spectral} to lie in a ``quartet'' of eigenvalues $\{\lambda,\lambda^*,-\lambda,-\lambda^*\}$, which may have one, two or four distinct elements, depending on the value of $\lambda$. To see this in the present context, first note that TRS of $h$ guarantees that $\mathcal{L}$ is real and therefore that $\lambda^*$ is also an eigenvalue, while acting with $J^{-1}$ yields
\begin{equation}
-\mathcal{L}J(J^{-1}\vec{v}) = -\lambda(J^{-1}\vec{v}),
\end{equation}
so that 
\begin{equation}
\mathrm{det}(-\mathcal{L}J - (-\lambda) \mathbb{1}) = \mathrm{det}(-J^T \mathcal{L}^T - (-\lambda) \mathbb{1}) = 0.
\end{equation}
Then the fact that $J^T = - J$, combined with the observation that $\mathcal{L}^T = \mathcal{L}$ by our assumptions of both Hermiticity and reality for $h$, implies
\begin{equation}
\mathrm{det}(J\mathcal{L} - (- \lambda) \mathbb{1}) = 0.
\end{equation}
We deduce that either the stationary point is linearly unstable or all eigenvalues of $J\mathcal{L}$ lie on the imaginary axis.

For real discrete breathers of the type defined above, we have $\bar{p}_a=0, \, \bar{q}_a = \sqrt{2}\phi_a$, yielding the simplification
\begin{equation}
\label{eq:Ldef}
\mathcal{L}_{ab} = \begin{pmatrix} \tilde{h}_{ab}  + g\phi_a^2 \delta_{ab} & 0 \\ 0 & \tilde{h}_{ab}  + 3 g\phi_a^2\delta_{ab}
\end{pmatrix} = \begin{pmatrix} (\mathcal{L}_{-})_{ab} & 0 \\ 0 & (\mathcal{L}_{+})_{ab}
\end{pmatrix},
\end{equation}
where
\begin{equation}
(\mathcal{L}_\pm)_{ab} = \tilde{h}_{ab} + (2 \pm 1) g \phi_{a}^2 \delta_{ab}.
\end{equation}
Thus the stability matrix reads
\begin{equation}
\label{eq: stabilityham}
J\mathcal{L} = \begin{pmatrix} 0 & -\mathcal{L}_+\\ \mathcal{L}_- & 0 \end{pmatrix}.
\end{equation}
In the present context, the matrix $J\mathcal{L}$ further has an exact zero eigenvalue, which is a consequence of the $U(1)$ symmetry of the dynamics. To see this, observe that each independent generator of a continuous symmetry of Eq.~\eqref{eq:rotated_frame} will yield a stationary solution to the linearized dynamics Eq.~\eqref{eq: stabilityham}. For the systems under consideration, the only continuous symmetry is a global $U(1)$ symmetry, whose Lie algebra $i\mathbb{R}$ acts on the target space as
\begin{equation} \label{eq: u1algebra_action}
\delta_{\epsilon} \tilde{\Phi}_{a}  = i\epsilon \tilde{\Phi}_{a}.
\end{equation}
At a given stationary point, $\tilde{\Phi} = \frac{1}{\sqrt{2}}(\bar{q}+i\bar{p})$, this flow is generated by the tangent vector
\begin{equation}
\begin{pmatrix} \delta p_a \\ \delta q_a \end{pmatrix} = \begin{pmatrix} \bar{q}_a \\ -\bar{p}_a \end{pmatrix}.
\end{equation}
Since this flow commutes with the time evolution Eq. ~\eqref{eq:rotated_frame}, it must always lie in the kernel of $J\mathcal{L}$. Specializing to the discrete breather solution $\tilde{\Phi}_a=\phi_a \in \mathbb{R}$, we deduce that the vector
\begin{equation}
\begin{pmatrix} \delta p_a \\ \delta q_a \end{pmatrix} = \begin{pmatrix} \phi_a \\ 0 \end{pmatrix}
\end{equation}
has stability eigenvalue zero. This also follows directly from Eqs.~\eqref{eq:time_indep_DNLSE} and \eqref{eq:Ldef}, upon noting that the former can be written as $\mathcal{L}_- \phi = 0$. 

\subsection{Perturbing away from solvable points} 
\label{subec:perturbingaway}
The spectrum of $J\mathcal{L}$ determines whether a given stationary point is linearly stable or unstable with respect to perturbations in phase space. In the present context an additional difficulty arises: discrete breather states solving Eq. \eqref{eq:time_indep_DNLSE} may be difficult to obtain analytically or even numerically for general values of $h_{ab}$ and $g$. In the spirit of MacKay and Aubry's ``anti-continuous limit''~\cite{MacKay1994,Aubry}, we can instead start from an analytically tractable combination of hopping matrix $h^{(0)}$ and interaction strength $g^{(0)}$, at which a breather solution $\phi^{(0)}$ to Eq. \eqref{eq:time_indep_DNLSE} both exists and is linearly stable. We would then like to argue either analytically (if possible) or numerically for existence and linear stability of breathers for some non-zero range of model parameters away from this point. 

To address this, let us restrict our attention to model parameters $(h,g)$ in some neighbourhood of $(h^{(0)},g^{(0)})$, in which a discrete breather solution $\phi$ both exists and is continuous in the model parameters. This implies that the coefficients in the characteristic polynomial of $J\mathcal{L}$ are continuous in the model parameters, and therefore that the eigenvalues of $J\mathcal{L}$ vary continuously in the complex plane. Due to the requirement that eigenvalues of $J\mathcal{L}$ form ``quartets'' in the complex plane, the only way in which unstable eigenvalues can arise is if two or more conjugate pairs of eigenvalues on the imaginary axis coalesce and move off the imaginary axis. Thus linear stability is guaranteed on neighbourhoods of $(h^{(0)},g^{(0)})$ in which no such bifurcations can occur.

The capacity of imaginary eigenvalues of $J\mathcal{L}$ to bifurcate is captured by the properties of their so-called \emph{Krein matrix}~\cite{kapitula2013spectral}, which to the best of our knowledge was first applied to breather physics in the foundational paper of MacKay and Aubry~\cite{MacKay1994}. We now summarize the parts of this approach that are used in the main text.

\subsection{Krein matrix and eigenvalues}
\label{sec: kreinsignature}
The basic result in the theory of the Krein matrix is the so-called ``Krein dichotomy''~\cite{kapitula2013spectral}: for any stability eigenvector $J\mathcal{L}\vec{v} = \lambda \vec{v}$, 
\begin{equation}
\lambda \notin i\mathbb{R} \implies \vec{v}^\dagger \mathcal{L} \vec{v} = 0.
\end{equation}
The Krein dichotomy follows directly from the identity
\begin{equation}
(\lambda + \lambda^*) \vec{v}^\dagger \mathcal{L} \vec{v} = 0,
\end{equation}
which can be deduced from the definition of $\vec{v}$, Eq.~\eqref{eq: v_definition}, in conjunction with the Hermiticity properties $J^\dagger = - J, \, \mathcal{L}^\dagger = \mathcal{L}$.

To motivate the Krein matrix in the present context, we assume that there exists a $P$-parameter family of deformations of the Hamiltonian $\tilde{H}(\epsilon)$, such that the discrete breather solution $\phi(\epsilon)$ exists and is continuous in $\epsilon$ in some neighbourhood of $\epsilon = 0 \in \mathbb{R}^P$. We are interested in the case that the breather solution is linearly stable at $\epsilon = 0$, and would like to understand when this stability persists for $\epsilon \neq 0$.

In order for such stability to persist, it is sufficient that none of the degenerate eigenvalues $\lambda$ at $\epsilon=0$ bifurcate. Let us suppose that among the stability eigenvalues of $\phi(\epsilon)$, there are $d$ stability eigenvalues $\lambda^{(1)}(\epsilon),\lambda^{(2)}(\epsilon),\ldots \lambda^{(d)}(\epsilon)$ with associated eigenvectors $\vec{v}^{(1)}(\epsilon),\vec{v}^{(2)}(\epsilon),\ldots,\vec{v}^{(d)}(\epsilon)$, that coalesce at $\epsilon = 0$, \emph{i.e.}, $\lambda^{(j)}(0) = \lambda \in i\mathbb{R}$ for all $j=1,2,\ldots,d$, with $\mathrm{Im}[\lambda]>0$. (The case $\lambda=0$ either arises due to the $U(1)$ symmetry in Eq.~\ref{eq: u1algebra_action}, in which case the corresponding eigenvalues are pinned at zero and cannot bifurcate for any $\epsilon$, or must be analyzed using different methods~\cite{kapitula2013spectral}.)

We further assume that the algebraic and geometric multiplicity of $\lambda$ are both equal to $d$. We then have the following Proposition, which summarizes the well-known results of Krein~\cite{kapitula2013spectral,CHERNYAVSKY201848} in a simple and self-contained form that is sufficient for our purposes below:
\begin{prop}
\label{prop1}
Let $\lambda^{(1)},\lambda^{(2)},\ldots,\lambda^{(d)}$ be as above, and define the Hermitian ``Krein matrix''
\begin{equation}
\label{eq: kreinmatrix}
K_{ij} = \vec{v}^{(i)}(0)^\dagger \mathcal{L}(0)\vec{v}^{(j)}(0)
\end{equation}
from their eigenvectors. Then if $K$ is either positive definite or negative definite, there is an open ball $\|\epsilon\| < r_\lambda$ with $r_\lambda > 0$ in which $\lambda$ does not bifurcate away from the imaginary axis.
\end{prop}
\begin{proof}
Suppose for a contradiction that $K$ is either positive or negative definite, and that in all open balls about $\epsilon=0$, there is a parameter value at which $\lambda$ bifurcates. Let $B(0,r)$ be such an open ball with radius $r>0$. Then there must be some even number $2 \leq 2m \leq d$ and a parameter value $\epsilon_r \in B(0,r)$ such that $2m$ of the eigenvalues $\lambda^{(j)}(\epsilon)$ lie off the imaginary axis when $\epsilon = \epsilon_r$. For sufficiently small $r$, which we henceforth assume, the indices $j$ labelling these ``bad'' eigenvalues will be independent of $r$, by continuity of the spectrum in $\epsilon$. Let us relabel the indices of these bad eigenvalues as $j=1,2,\ldots,2m$. From the corresponding eigenvectors, we define the $2m \times 2m$ matrix
\begin{equation}
K'_{ij}(\epsilon) = \vec{v}^{(i)}(\epsilon)^\dagger \mathcal{L}(\epsilon) \vec{v}^{(j)}(\epsilon), \quad i,j=1,2,\ldots,2m,
\end{equation}
which is a square, Hermitian submatrix of $K$ when $\epsilon=0$. Now by assumption, $\lambda^{(j)} (\epsilon_r) \notin i \mathbb{R}$ for $j=1,2,\ldots 2m$, and therefore by the Krein dichotomy, 
\begin{equation}
\mathrm{Tr}[K'(\epsilon_r)] =0.
\end{equation}
As $r \to 0$, this must hold for arbitrarily small values of $0 < \|\epsilon_r\| <r$, and by continuity of the spectrum in $\epsilon$ it follows that
\begin{equation}
\mathrm{Tr}[K'(0)] = 0.
\end{equation}
But at $\epsilon=0$, $K'$ is a submatrix of $K$ and therefore has non-zero trace by definiteness of $K$, which is a contradiction.
\end{proof}

In this paper, we refer to the eigenvalues of the Krein matrix $K_{ij}$ in Proposition~\ref{prop1} as the \emph{Krein eigenvalues} $\gamma^{(i)}$, $i = 1 \dots d$. The condition of positive or negative definiteness of the Krein matrix is then equivalent to all the Krein eigenvalues having the same sign.

Note that if there are multiple sets of degenerate eigenvalues $\lambda$ at $\epsilon=0$ satisfying the hypotheses of Proposition~\ref{prop1}, each of these has an associated ``radius of linear stability'' $r_{\lambda} > 0$. The radius of linear stability of the discrete breather solution is then $r = \inf_{\lambda} r_{\lambda}$. If $J\mathcal{L}(0)$ has infinitely many distinct eigenvalues, the possibility arises that $r=0$. However, whenever $J\mathcal{L}$ has finitely many distinct eigenvalues satisfying the hypotheses of Proposition \ref{prop1}, then the radius of linear stability $r > 0$. This is the case for all the models studied in this paper.

\section{Review of the DNLSE in 1D}
\label{appendix: dnlse}

To illustrate the concepts outlined above, we start with the simple and well-studied example of the discrete nonlinear Schr{\"o}dinger equation (DNLSE) with nearest-neighbour hopping.

\subsection{Tight-binding matrix and symmetries}
We assume one site per unit cell, so that $a = R$, where $R = 1 \dots L$ are coordinates on a periodic lattice of length $L$ (such that $L +1 \equiv 1$), and  
\begin{equation} \label{eq:simple_tb_def}
h_{R R'} = -\epsilon (\delta_{R+1,R'} + \delta_{R-1,R'}),
\end{equation}
where $\epsilon > 0$ is an energy scale/hopping amplitude. (The internal index $\alpha$ is redundant here.)
In the electronic structure context, this model realizes the simplest example of an \emph{atomic} -- rather than molecular -- band, where the Wannier states of the band are localized directly on top of the lattice sites instead of being shared between them. The resulting evolution equation Eq.~\eqref{eq:nonlinear_tightbinding} is the standard nonintegrable discretization~\cite{kevrekidis2001discrete} of the nonlinear Schr{\"o}dinger equation.

This model does not have crystalline symmetries beyond translations and the breather symmetry data is therefore trivial.

\subsection{Molecular breather limit}
The molecular limit of the DNLSE is obtained by setting $\epsilon = 0$. For this model, the molecular limit is precisely MacKay and Aubry's anti-continuous limit~\cite{MacKay1994,Aubry}, and existence and stability of discrete breathers for some range of $\epsilon \neq 0$ is well established~\cite{MacKay1994,johansson1997existence,Aubry}. It will nevertheless be useful to review the key concepts below.

\subsubsection{On-site breather mode}
A maximally-localized breather solution of the form Eq.~\eqref{eq: mostlocalized_soln} reads
\begin{equation} \label{eq: dnlse_molecularlimitsoln}
\phi^\mathrm{I}_R = \delta_{R,0}, \quad \omega = g.
\end{equation}
Let us first consider existence of discrete breathers for $\epsilon > 0$. This is determined by anharmonicity of the dynamics in the anti-continuous limit~\cite{johansson1997existence}. To show this, we first note that the DNLSE Hamiltonian can be written in action-angle coordinates $\Phi_{R} = I_R^{1/2}e^{-i\theta_{R}}$ as
\begin{equation}
H = \sum_{R=1}^{L} -2\epsilon(I_RI_{R+1})^{1/2}\cos{(\theta_R-\theta_{R+1})} + \frac{g}{2}I_R^2.
\end{equation}
Thus the frequency at $R=0$ in the anti-continuous limit is given by
\begin{equation}
\omega(I_0) = gI_0,
\end{equation}
so that $\partial_{I_0}\omega(I_0) = g$, and the breather is anharmonic by our assumption that $g \neq 0$. Existence of discrete breathers follows for some range of $\epsilon \neq 0$ by the theorem of MacKay and Aubry~\cite{MacKay1994}.

\begin{figure}[t]
\centering
\includegraphics[width=0.48\textwidth,page=1]{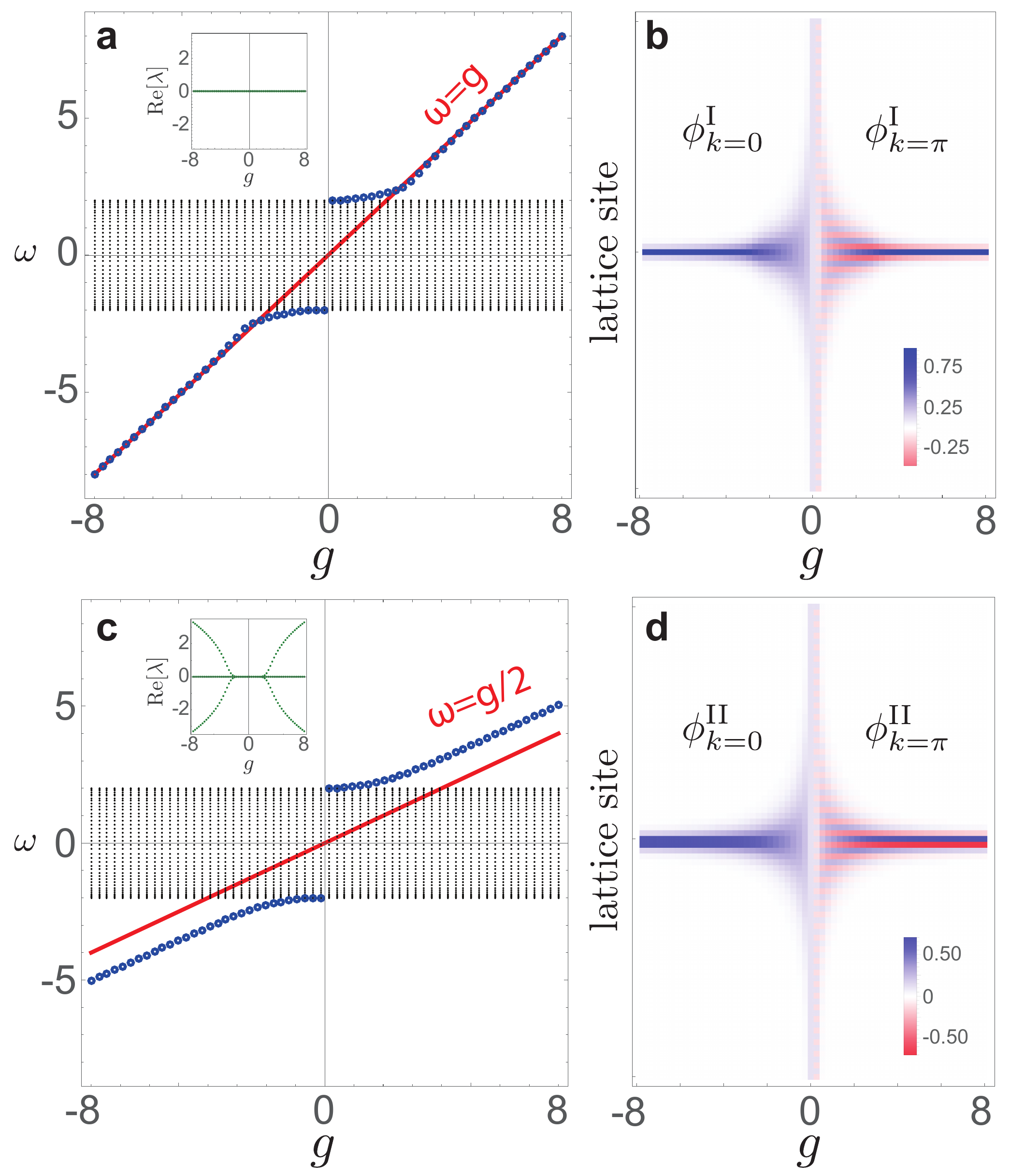}
\caption{Numerical breather solutions in the DNLSE (Eqs.~\eqref{eq:nonlinear_tightbinding}, \eqref{eq:simple_tb_def}) with $\epsilon=1$. \textbf{a}~Linear spectrum (black) and one-site breather energies (blue). The inset shows the real part of the stability eigenvalues $\textrm{Re}[\lambda]$. \textbf{b}~Breather profiles $\phi_R$. Blue (red) represents positive (negative) components. \textbf{c, d}~Same as \textbf{a, b} but for two-site breathers.
}
\label{fig:numerics_metal}
\end{figure}

We next turn to the question of stability. We first note that the breather site contribution to the stability matrix, Eq. \eqref{eq:stabhambreath}, reads
\begin{equation}
    B_0 =  \begin{pmatrix} 0 & 
    -2g
    \\ 
    0
    & 0 \end{pmatrix}.
\end{equation}
This matrix is defective, with a stability eigenvalue $\lambda=0$ that has algebraic multiplicity $2$ and geometric multiplicity $1$. The sites with $R \neq 0$ each contribute the block Eq. \eqref{eq:stabhamnonbreath}, given here by
\begin{equation}
B_1 = \begin{pmatrix} 0 & 
    g
    \\ 
    -g
    & 0 \end{pmatrix},
\end{equation}
with spectrum $\lambda = \pm i g$. Thus the stability spectrum in the anti-continuous limit consists of two $(L-1)$-fold degenerate eigenvalues $\pm ig$ and a single doubly degenerate eigenvalue $\lambda=0$. It follows that the discrete breather solution is (marginally) linearly stable in the anti-continuous limit.

Provided $g \neq 0$, we deduce from the discussion in Sec. \ref{subsec: molecular_stability} that there exists a non-zero neighbourhood of $\epsilon = 0$ in which the stability eigenvalues $\lambda = \pm ig$ cannot bifurcate. Meanwhile, since one stability eigenvalue $\lambda$ is always pinned at $0$ by symmetry, its companion cannot bifurcate for sufficiently small $|\epsilon|$ while preserving the constraint that stability eigenvalues must come in time-reversed pairs $\pm \lambda$. We deduce that linear stability persists for some range of $\epsilon \neq 0$.

\subsubsection{Two-site breather mode}
Next, we construct two off-site breathers localized at the boundary of the unit cell (whose nominal midpoint lies at a fractional lattice coordinate ``$R = 1/2$"). These are the next most localized breather modes after on-site breathers as in Eq.~\eqref{eq: dnlse_molecularlimitsoln}, and it is easily verified that the choice
\begin{equation}\label{eq: dnlse_molecularlimitsoln2}
\phi^\mathrm{II}_R = \frac{1}{\sqrt{2}} \left(\delta_{R0} \pm \delta_{R1}\right), \quad \omega' = \frac{g}{2}
\end{equation}
satisfies Eq. \eqref{eq:time_indep_DNLSE}. 
For two-site breathers, the stability eigenvalues still all lie on the imaginary axis. However, the support of the breather now contributes four (algebraic) zero eigenvalues to the spectrum of $J\mathcal{L}$. Thus there is the possibility of bifurcation and linear instability when $\epsilon \neq 0$
(note that Krein's criterion does not apply to zero eigenvalues~\cite{kapitula2013spectral}).

\subsection{Perturbing away from the molecular limit}
To complete our overview of the simplest case of the DNLSE, we now solve numerically for breathers away from the anti-continuous limit. The results are summarized in Fig.~\ref{fig:numerics_metal}. For $|g|<|\epsilon|$, the breathers are somewhat delocalized as they bifurcate from the linear bands (black points), while for $|g|>|\epsilon|$ they increasingly resemble the analytic expressions of Eqs.~\ref{eq: dnlse_molecularlimitsoln} and \ref{eq: dnlse_molecularlimitsoln2} with increasing $|g|$. The lower (upper) band edge of the linear spectrum corresponds to Bloch states with wave vector $k=0$ ($k=\pi$); correspondingly, breathers with $\omega<0$ ($\omega>0$) have components in phase (out of phase) as they bifurcate from this band edge. While the values of $\omega$ for $\phi^\mathrm{I}_R$ breathers converge to the anti-continuous limit predictions (red line) as $|g|$ is increased, those for $\phi^\mathrm{II}_R$ breathers are shifted by order one numbers $\pm \epsilon$, which reflect the kinetic energy associated with hopping between the occupied sites at $R=0$ and $R=1$.

\section{Nonlinear symmetry constraints} \label{appendix: symmetries}
In this Appendix, we discuss unitary symmetries of discrete breathers in nonlinear-Schr{\"o}dinger-type systems.

\subsection{General unitary transformations}
A unitary transformation $\Phi_a = \sum_{\tilde{a}} U_{a \tilde{a}} \Psi_{\tilde{a}}$, where $U^\dagger U = 1$, maps Eq.~\eqref{eq:nonlinear_tightbinding} to
\begin{equation} \label{syseq:DNLSEnewbasis}
\begin{aligned}
i \frac{\mathrm{d}}{\mathrm{d}t} \Psi_{\tilde{a}} =& \sum_{\tilde{b}} \left(\sum_{a b} h_{a b} U^*_{a \tilde{a}} U_{b \tilde{b}} \right) \Psi_{\tilde{b}} \\&+ g \sum_{\tilde{b} \tilde{c} \tilde{d}} \left(\sum_a U^*_{a \tilde{a}} U^*_{a \tilde{b}} U_{a \tilde{c}} U_{a \tilde{d}} \right) \Psi^*_{\tilde{b}} \Psi_{\tilde{c}} \Psi_{\tilde{d}}.
\end{aligned}
\end{equation}
Here, we use a tilde for all transformed coordinate indices because they may not correspond to the lattice basis. We deduce that any unitary symmetry of the dynamics must satisfy two requirements: (1) it must commute with the hopping matrix: $[U,h]=0$, and (2) it must satisfy
\begin{equation} \label{syseq:NonlinearSymmetryConstraint}
\sum_a U^*_{a \tilde{a}} U^*_{a \tilde{b}} U_{a \tilde{c}} U_{a \tilde{d}} = \delta_{\tilde{a} \tilde{b}} \delta_{\tilde{a} \tilde{c}} \delta_{\tilde{a} \tilde{d}}.
\end{equation}
The second constraint depends on the specific form of the on-site nonlinearity in Eq.~\eqref{eq:nonlinear_tightbinding}, and guarantees that this nonlinearity will be an on-site term in the transformed basis.

\subsection{Point group symmetries}
To describe a point group symmetry operation $S$, we write $a=(\bs{R},\alpha)$ where $\bs{R}$ is the unit cell coordinate and $\alpha$ labels the orbitals that transform under the representation $\rho$ of $S$ with matrix elements $\rho^{S}_{\alpha \alpha'}$. For simplicity, we assume that all orbitals are located at the unit cell center. Since point group symmetries map the lattice basis to itself, we may write $\tilde{a} = (\bs{R}',\alpha')$. The point group operation about the origin $\bs{R}=\bs{0}$ takes the form
\begin{equation}
U^S_{(\bs{R},\alpha),(\bs{R}',\alpha')} = \rho^{S}_{\alpha \alpha'} \delta_{\bs{R},S\bs{R}'},
\end{equation}
where $S\bs{R}'$ is the lattice point that results from acting with $S$ on $\bs{R}'$, correspondingly, we are using the symbol $S$ for both the abstract symmetry and its matrix representation in real space. (In Sec.~\ref{sec: molecularlimit} of the main text, we had abbreviated $U^S \equiv U$.) 
\begin{widetext}
We then find
\begin{equation}
\begin{aligned}
\sum_{\bs{R},\alpha} &U^{S*}_{(\bs{R},\alpha),(\bs{R}',\alpha')} U^{S*}_{(\bs{R},\alpha),(\bs{R}'',\alpha'')} U^{S}_{(\bs{R},\alpha),(\bs{R}''',\alpha''')} U^{S}_{(\bs{R},\alpha),(\bs{R}'''',\alpha'''')} \\&= \sum_{\bs{R},\alpha} \rho^{*S}_{\alpha \alpha'}\delta_{\bs{R},S\bs{R}'} \rho^{*S}_{\alpha \alpha''} \delta_{\bs{R},S\bs{R}''} \rho^{S}_{\alpha \alpha'''} \delta_{\bs{R},S\bs{R}'''} \rho^{S}_{\alpha \alpha''''} \delta_{\bs{R},S\bs{R}''''} 
= \left(\sum_{\alpha} 
\rho^{*S}_{\alpha \alpha'} \rho^{*S}_{\alpha \alpha''} \rho^{S}_{\alpha \alpha'''} \rho^{S}_{\alpha \alpha''''} \right) \delta_{\bs{R}',\bs{R}''} \delta_{\bs{R}',\bs{R}'''} \delta_{\bs{R}',\bs{R}''''},
\end{aligned}
\end{equation}
\end{widetext}
so that Eq.~\eqref{syseq:NonlinearSymmetryConstraint} is satisfied if and only if 
\begin{equation} \label{syseq:sym_condition_on_pgrep}
\sum_{\alpha} 
\rho^{*S}_{\alpha \alpha'} \rho^{*S}_{\alpha \alpha''} \rho^{S}_{\alpha \alpha'''} \rho^{S}_{\alpha \alpha''''} = \delta_{\alpha' \alpha''} \delta_{\alpha' \alpha'''} \delta_{\alpha' \alpha''''}.
\end{equation}
For a diagonal representation of $S$ we have $\rho^{S}_{\alpha \alpha'} = \lambda_{\alpha} \delta_{\alpha \alpha'}$ where the orbitals have eigenvalues $\lambda_{\alpha}$ with $|\lambda_{\alpha}|^2=1$. In this case, Eq.~\eqref{syseq:sym_condition_on_pgrep} is trivially fulfilled. The same is true for all models in the main text where $S$ -- in the guise of e.g. inversion symmetry or $\mathcal{C}_3$ symmetry -- permutes the orbitals within the unit cell, that is,
\begin{equation}
    \rho^S_{\alpha \alpha'} = \delta_{\alpha, (\alpha'+1 \, \mathrm{mod}\, n)},
\end{equation}
where $n$ is the order of the point group symmetry such that $S^n = \mathbb{1}$.

\section{Numerical methods} 
\label{appendix: numerics}
Here, we describe the numerical algorithm that we have employed to find self-consistent breather solutions to Eq.~\eqref{eq:time_indep_DNLSE}. The method consists of iteratively diagonalizing an \emph{effective linear Hamiltonian}, defined herein. Consider a collection of breather fields $\phi^{(i)}$, where $i$ labels the iteration. Let $\phi^{(i=0)}$ be an ansatz for the breather, chosen from the exact breather solutions found in the molecular limits discussed in the main text, as well as the pseudomomentum $k$ of the Bloch state from which the breather bifurcates (if it does). We also set the initial effective linear Hamiltonian to be the linear Hamiltonian, $h^{(i=0)}=h$ (see Eq.~\ref{eq:nonlinear_tightbinding}). We obtain an effective linear Hamiltonian for the next iteration by computing
\begin{align}
    h^{(i+1)}=h^{(i)}+g |\phi^{(i)}|^2
\end{align}
which we then diagonalize,
\begin{align}
    h_{ab}^{(i+1)}=\sum_j\psi^{j}_a E_j \psi^{j*}_b,
\end{align}
where $a$, $b$ are the composite indices $a,b=({\bf R}, \alpha)$ running over unit cells labelled by ${\bf R}$ and their internal degrees of freedom $\alpha$.
Now, consider the set of overlaps between the current breather field $\phi^{(i)}$ and the $j\mathrm{th}$ eigenstate of $h^{(i+1)}$ given by $o_j=|\sum_a \phi^{(i)*}_a \psi^j_a|^2$, for $j=1\ldots N$, where $N$ is the dimension of $h$. We set the new breather field to be the eigenstate of $h^{(i+1)}$ with the largest overlap $o_j$, \emph{i.e.}, we set $\phi^{(i+1)}=\psi^J$, such that $o_J \geq o_j$ for all $j \in 1 \ldots N$.

The iterations can continue until the overlap $o_j$ approaches a value close to 1 or until a maximum number of iterations is reached. We have chosen the second criterion because the first one often results in unpredictable transitions to linearly stable solutions.

\section{Linear stability theory for specific models}
In this Appendix, we summarize the linear stability theory for the specific solvable examples of molecular breathers considered in the main text. 
\subsection{Nonlinear SSH chain}
\label{app:SSHstab}
In the molecular limit of the SSH model, the spectrum of $h_{\alpha\beta}$ is given by $E^{(1)} = -E^{(2)} = m$. Thus the spectrum of the block $B_1$ in Eq. \eqref{eq:stabhamnonbreath} is $\lambda = \pm i(\pm m - \omega)$ and the resonance condition Eq. \eqref{eq: mysterycondition} becomes $\omega = 0$.

Let us first focus on the symmetry-breaking breathers $(\mathrm{I},\pm)$, which turn out to have the same stability eigenvalues. Calculating the spectrum of $B_0^2$ (see Eq. \eqref{eq:stabhambreath}) and the allowed structure of eigenvalues of $B_0$ in the complex plane (see Sec. \ref{sec:linear_stability}) implies that $B_0$ has a pair of eigenvalues $\lambda = \pm i \sqrt{g^2-4m^2}$ and a pair of zero eigenvalues. Since the block $B_0$ is the only possible source of linear instabilities, these breathers are linearly stable on their entire domain of existence $|g| \geq 2m$. The spectrum of the block $B_1$ in Eq. \eqref{eq:stabhamnonbreath} is given by $\lambda = \pm i(\pm m - g)$. Thus the non-zero stability spectrum of the symmetry-breaking breathers consists of six generically distinct imaginary eigenvalues
\begin{equation}
\lambda = \begin{cases} \pm i \sqrt{g^2-4m^2} \\ \pm i (g-m) \\ \pm i (g+m)
\end{cases}.
\end{equation}
For $|g| > 2m$ there are precisely two collisions between eigenvalues on the positive imaginary axis, at $g = \pm 5m/2$, where a $B_0$ eigenvalue intersects a $B_1$ eigenvalue. However, as shown in Fig.~\ref{fig: ssh_theory}b, the corresponding Krein eigenvalues are of the same sign, which prevents linear instability for sufficiently small inter-cell hopping $\epsilon \neq 0$. For all other values of $|g| > 2m$, linear stability is guaranteed for some range of $\epsilon \neq 0$, since the $B_0$ eigenvalues are prevented from bifurcating by their multiplicity or by $U(1)$ symmetry, and the $B_1$ eigenvalues are prevented from bifurcating because the resonance condition Eq. \eqref{eq: mysterycondition} is never satisfied.

We next consider the symmetric breathers $(\mathrm{II},\pm)$. First consider the inversion-even breather $(\mathrm{II},+)$. The spectrum of $B_0$ yields a pair of eigenvalues $\lambda = \pm i \sqrt{2m(2m+g)}$, and a pair of (algebraic) zero eigenvalues. Thus there is a linear instability for $g < - 2m$. Away from level crossings with $B_1$, neither of these pairs can bifurcate for $\epsilon \neq 0$ sufficiently small as above, unless $g = -2m$ where $B_0$ has a quartet of zero eigenvalues. The spectrum of $B_1$ is given by $\lambda = \pm i (\pm m + m - g/2)$. This is non-zero and imaginary, unless $g=0$ or $g=4m$. We deduce that the non-zero stability spectrum of the inversion-even breather consists of six distinct imaginary values
\begin{equation}
\lambda = \begin{cases} \pm i \sqrt{2m(2m+g)} \\
\pm i (g/2-2m) \\
\pm i g/2
\end{cases},
\end{equation}
in its region of stability $g \geq - 2m$, with possible bifurcations from $\lambda=0$ at $g=-2m$, $g=0$ and $g=4m$, and four collisions between $B_0$ and $B_1$ eigenvalues on the positive imaginary axis. As shown in Fig.~\ref{fig: ssh_theory}d, these all have Krein eigenvalues of the same sign. In addition, the resonance condition Eq. \eqref{eq: mysterycondition} is satisfied at $g = 2m$, allowing bifurcations to occur for $\epsilon \neq 0$. For all other values of $g \geq -2m$, the inversion-even breather is guaranteed to be stable in some neighbourhood of $\epsilon=0$.

Next consider the inversion-odd breather $(\mathrm{II},-)$. The analysis proceeds much as for the inversion-even breather, and we find the stability spectrum
\begin{equation}
\lambda = \begin{cases} \pm i \sqrt{2m(2m-g)} \\
\pm i (g/2+2m) \\
\pm i g/2
\end{cases},
\end{equation}
which is related to the inversion-even stability spectrum by a reflection $g \mapsto -g$. The resonance condition Eq. \eqref{eq: mysterycondition} is now satisfied at $g = -2m$, and there are possible bifurcations from $\lambda = 0$ at $g=-4m$, $g=0$ and $g=2m$. We deduce that the inversion-odd breather is linearly stable for $g \leq 2m$, and for all $g$ in this range outside a measure zero set of values, linearly stable for some range of $\epsilon \neq 0$.

\subsection{Nonlinear breathing Kagome lattice}
\label{app:BKLstab}
The spectrum of $h_{\alpha\beta}$ for the molecular limit of this 2D breathing kagome lattice is given by $E^{(1)}=m, \, E^{(2)}=-2m$, with the former eigenvalue having algebraic multiplicity two. Thus the spectrum of the block $B_1$ in Eq. \eqref{eq:stabhamnonbreath} is given by $\lambda \in \{\pm i (m-\omega), \pm i(-2m-\omega)\}$, and the resonance condition Eq. \eqref{eq: mysterycondition} is $\omega = -m/2$.

Since the breather of Type II is complicated to write down analytically, we perform its stability analysis numerically.

For breathers of Type III, the spectrum of the block $B_1$ is given by
$\lambda \in \{\pm ig/2, \pm i(3m+g/2)\}$, while the spectrum of block $B_0$ is given by $\lambda \in \{ \pm i \sqrt{(\gamma \pm \sqrt{\gamma^2+2mg^3})/2}\}$ where $\gamma = 9m^2-mg+g^2/4$. Thus the stability spectrum is given by
\begin{equation}
\lambda = \begin{cases}\pm i \sqrt{(\gamma \pm \sqrt{\gamma^2+2mg^3})/2} \\
\pm i g/2\\
\pm i (3m+g/2).
\end{cases}
\end{equation}
We find numerically that the left boundary of the region of stability where $\gamma^2 + 2mg^3 = 0$ occurs at $g \approx -9.07$, while the right boundary is easily seen to lie at $g = 0$. Thus the region of linear stability for Type III breathers is given by $-9.07m \lesssim g \leq 0$. In the interior of this region the spectrum of $B_0$ is at most two-fold degenerate, while the only possible bifurcation due to the spectrum of $B_1$ crossing zero occurs when $g = -6m$. Around $g \approx -3.5m$, there is also a single value of $g$ at which the $B_1$ eigenvalue $\pm i (3m+g/2)$ crosses the non-zero $B_0$ eigenvalue. As shown in Fig.~\ref{fig: oai_theory}d, the respective Krein eigenvalues are of opposite sign, so that this crossing may turn into an instability away from the molecular limit. Finally, the resonance condition Eq. \eqref{eq: mysterycondition} is satisfied at $g = -3m$, again allowing for linear instability when $\epsilon \neq 0$. Away from these three points, the Type III breather is expected to be stable for some range of $\epsilon \neq 0$.

For breathers of Type IV, the block $B_0$ has non-zero eigenvalues  $\lambda = \pm i \sqrt{9m^2 + 2mg}$, while the spectrum of the block $B_1$ is given by
$\lambda \in \{\pm i g/3, \pm i (3m-g/3)\}$. Thus the non-zero stability spectrum reads
\begin{equation}
\lambda = \begin{cases} \pm i \sqrt{m(9m+2g)} \\
\pm i g/3\\
\pm i (3m-g/3)
\end{cases},
\end{equation}
implying that Type IV breathers are linearly stable for $g \geq -9m/2$. The resonance condition Eq. \eqref{eq: mysterycondition} is satisfied at $g = 9m/2$ allowing for possible bifurcations for $\epsilon \neq 0$.

The dispersion of these solutions with $g$, their linear stability eigenvalues (Sec.~\ref{sec:linear_stability}), and their Krein eigenvalues (Sec.~\ref{sec: kreinsignature}) are summarized in Fig.~\ref{fig: oai_theory}. We find that the solutions of Type I are only stable in the range $g\lesssim -4m$ and $g>0$ where the frequency of these solutions lies outside the band gap. In these regimes, the Krein eigenvalues have the same sign for any given $g$ implying stability away from the molecular limit. Conversely, solution II is stable for $g\gtrsim -4.5m$, both outside and between the linear bands, but has stability bands with opposite Krein eigenvalue that cross at $g=4.5m$ where condition Eq.~\ref{eq: mysterycondition} is fulfilled. Correspondingly, for $g \approx 4.5m$ and small $\epsilon \neq 0$, this solution may develop an instability away from the molecular limit, which we investigate numerically in Sec.~\ref{sec: c3numerics}. Finally, the solutions of Type III are only stable within the band gap and in the range $-9.07m \lesssim g \leq 0$; for $g \approx -3m$ and small $\epsilon \neq 0$ they may become unstable away from the molecular limit due to a crossing between stability bands of opposite Krein eigenvalue (where again condition Eq.~\ref{eq: mysterycondition} is fulfilled).

\subsection{Testing the stability of breathers in the nonlinear breathing Kagome lattice}
\label{app:StabTest2D}
In this section, we numerically test the stability of the breather solutions in the nonlinear Kagome lattice with real-space hopping matrix~\eqref{eq: c3model_realspacehamiltonian}. Figure~\ref{fig:numerics_Kagome_Dynamics} shows the time evolution of all the breather solutions over $10^4$ periods of oscillation $T$ of each breather. In (a), $g=-12$. At this value of $g$, the linear stability analysis in the main text predicts a stable breather $\phi^{(\mathrm{II})}$ and unstable breathers $\phi^{(\mathrm{III})}$ and $\phi^{(\mathrm{IV})}$ (see Fig.~\ref{fig: oai_theory}a, for summary of results). Although these predictions are linear in perturbation theory and only expected to hold in a sufficiently small neighborhood of the molecular limit, the time evolution of Fig.~\ref{fig:numerics_Kagome_Dynamics}(a) agrees with these predictions, as only the first plot shows a constant breather during time evolution. In (b), $g=-8$. At this value of $g$, the linear stability analysis in the main text predicts the stable breathers $\phi^{(\mathrm{II})}$ and $\phi^{(\mathrm{III})}$, and an unstable breather $\phi^{(\mathrm{IV})}$. The time evolution in Fig.~\ref{fig:numerics_Kagome_Dynamics}(b) agrees with these predictions. These results suggest that higher orders in perturbation theory do not appreciably change the linear stability of the breathers and that the predictions derived in the main text are robust to both nonlinear instabilities and to deviations from the molecular limit.
\begin{figure*}[t]
\centering
\includegraphics[width=\textwidth,page=1]{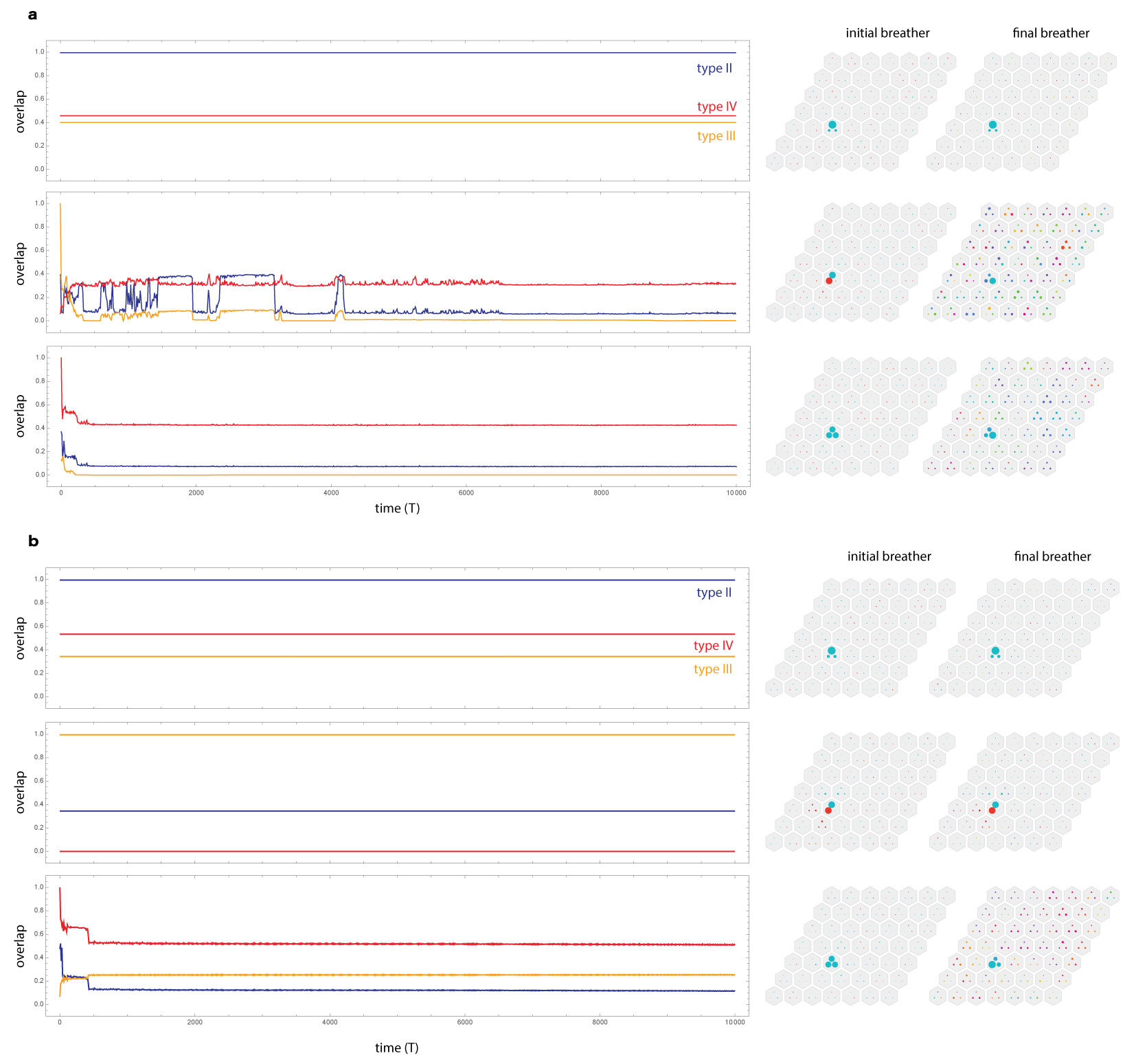}
\caption{Time-evolution of breathers in the nonlinear Kagome lattice with real-space hopping matrix~\eqref{eq: c3model_realspacehamiltonian} under small perturbations. The parameters of the lattice are $m=1$ and $\epsilon=0.1$. The strength of nonlinearity is set to $g=-12$ in (a) and $g=-8$ in (b). At these values of $g$, there are three breather solutions: $\phi^{(\mathrm{II})}$, $\phi^{(\mathrm{III})}$, and $\phi^{(\mathrm{IV})}$. The initial fields are $\phi^{(\mathrm{i})}(t=0)=\phi^{(\mathrm{i})}+\phi_\mathrm{pert}$, for $\mathrm{i=II}$, $\mathrm{III}$, and $\mathrm{IV}$, where $\phi_\mathrm{pert}$ is a random field drawn from a uniform distribution between $\pm 0.01$. The fields $\phi^{(\mathrm{i})}(t=0)$ are normalized to have unit magnitude, which is preserved throughout the time evolution. Each plot corresponds to one time evolution, and the blue, red, and yellow lines indicate the overlaps (inner product) of the instantaneous field with the unperturbed breathers $\phi^{(\mathrm{i})}$. The simulations were performed using the RK4 algorithm with a time step $\Delta t = 0.005$ and run for 10,000 breather periods.
}
\label{fig:numerics_Kagome_Dynamics}
\end{figure*}

\bibliography{refs}
\end{document}